\documentclass[11pt]{article}
\usepackage[top=25mm,bottom=25mm,left=25mm,right=25mm]{geometry}
\usepackage{amsmath,amssymb,amsthm}
\usepackage{graphicx}
\usepackage{booktabs,array,tabularx,longtable}
\usepackage{float}
\usepackage[hidelinks]{hyperref}
\usepackage[utf8]{inputenc}
\usepackage[protrusion=true,expansion=false]{microtype}
\usepackage{parskip,setspace,caption,authblk}
\newcommand{\Nstar}{N_{\star}}
\newcommand{\sigstar}{\sigma_{0}^{*}}
\newcommand{\thetabio}{\theta}
\newcommand{\thetahat}{\hat\theta}
\newtheorem{assumption}{Assumption}
\newtheorem{proposition}{Proposition}

\newtheorem{definition}{Definition}

\begin{document}

\title{\textbf{Thermodynamic Parametrisation of the Vertebrate Lifetime
Cycle Invariant: Biological Proper Time, Allometric
Mass-Cancellation, and Clade-Specific Predictions}}

\author[1,*]{Mesfin~Asfaw~Taye}
\affil[1]{West Los Angeles College, Science Division,
9000 Overland Ave, Culver City, CA 90230, USA}
\affil[*]{Correspondence: \texttt{tayem@wlac.edu}}
\date{}
\maketitle

\begin{abstract}
Warm-blooded vertebrates accumulate approximately
$\Nstar \approx 10^9$ cardiac cycles over a natural lifetime,
a striking empirical regularity first quantified by Lindstedt and
Calder yet lacking a physical interpretation.
We propose that this invariance is consistent with a conserved
thermodynamic budget, formulated here as the Principle of Biological
Time Equivalence (PBTE).
The framework rests on a constitutive closure
$\dot{\Sigma} = \sigma_0 f$, which links the entropy production rate
to the intrinsic physiological frequency; integration over the
lifespan yields $\Sigma_{\mathrm{life}} = \sigma_0 \Nstar$, so that
the observed constancy of $\Nstar$ corresponds to an approximately
constant lifetime entropy budget.
Algebraic exponent cancellation under Kleiber and Calder scaling laws,
$\sigstar \propto M^{3/4+1/4-1}=M^0$, is consistent with
mass-independence and reproduces the numerical value
$N_0 \approx 1.52\times10^9$ without free parameters.
The framework offers a thermodynamically consistent account of two
outstanding problems: the origin of the numerical value of $\Nstar$
and the systematic deviations observed across clades.
A multiplicative correction factor $\Phi_C$, constructed from
physiological determinants---activity allocation, body temperature,
mitochondrial efficiency, and extrinsic hazard---predicts long-lived
clades as regimes of reduced effective entropy production per cardiac
cycle.
Four clades are treated in detail.
Primates maintain elevated neural metabolic fractions
$\phi = P_{\rm brain}/P_{\rm body} \in [0.06, 0.20]$ that suppress
somatic entropy production through predictive homeostatic regulation,
enhanced cellular repair, and behavioural risk buffering, yielding a
neuro-metabolic multiplier $\Phi_{\rm neuro} = (\phi/\phi_0)^\alpha$
with $\alpha \approx 0.40$.
Bats exploit two simultaneous mechanisms during torpor: drastic
suppression of time-averaged cardiac frequency and hypothermic
depression of biochemical damage kinetics, whose product can exceed
$\Phi_{\rm bat} \approx 8$ for deeply hibernating vespertilionids.
Birds present a thermodynamic paradox in which elevated body
temperatures and flight-induced cardiac acceleration both act
adversely, yet exceptional mitochondrial coupling efficiency and
antioxidant capacity generate a biochemical factor
$\Phi_{\rm mito+oxid} \approx 2.33$ that overcomes both deficits.
Cetaceans accumulate the majority of their lifespans in bradycardic
dive states, producing a duty-cycle correction that raises the
effective thermodynamic budget by a factor of three above the naive
raw count.
Biological proper time, defined as the accumulated cycle count
$\theta_i(t)=\int_0^t f_i(t')\,\mathrm{d}t'$, furnishes a natural
state variable replacing chronological time in describing organismal
aging: it is the unique coordinate in which entropy accumulates
uniformly, developmental milestones occupy universal fractions of the
lifetime, and epigenetic clock rates become species-independent.
The formal structure of this coordinate---its geometric interpretation
as arc length, its transformation group, and its clinical extension
as a temporal order parameter for disease classification---is
developed in detail.
The central constitutive assumption, that the entropy cost per cardiac
cycle is approximately conserved across species, remains to be
directly tested by simultaneous calorimetric and cardiac telemetric
measurement.
Until such tests are performed, PBTE should be regarded as a
thermodynamically motivated, internally consistent parametrisation
that is empirically concordant across clades but not yet explanatory
in the strict physical sense.
\end{abstract}


\section{Introduction}
\label{sec:intro}

Time, in biology, presents itself simultaneously as both an external
parameter and an intrinsic resource.
As an external parameter, it serves as the independent variable against
which physiological states evolve; as an intrinsic resource, it represents
a finite capacity that every organism possesses, consumes, and ultimately
exhausts.
These two interpretations are ordinarily conflated, because chronological
time provides a convenient and operational proxy for both.
This conflation, however, becomes untenable when one attempts to compare
organisms whose metabolic rates, body sizes, and lifespans differ by many
orders of magnitude.

\medskip
\noindent
A canonical illustration is provided by the comparison between small and
large mammals.
A common shrew (\textit{Sorex araneus}, $\sim 8$\,g) maintains a resting
heart rate of approximately $800$\,bpm and survives for roughly eighteen
months.
An African elephant (\textit{Loxodonta africana}, $\sim 4{,}000$\,kg),
by contrast, exhibits a resting heart rate near $28$\,bpm and a lifespan
of approximately sixty-five years.
When viewed in chronological time, these organisms appear to inhabit
vastly different temporal domains.
Yet when expressed in terms of accumulated physiological cycles—the
heartbeats that drive oxygen transport, metabolic flux, and cellular
maintenance—their lifetimes are remarkably comparable.
The shrew accumulates approximately $6.3 \times 10^8$ cardiac cycles,
while the elephant accumulates approximately $10^9$.

\medskip
\noindent
This near-equivalence is not a coincidence arising from two extreme
examples, but rather reflects a broader biological regularity.
Rubner~\cite{rubner1908} first identified the approximate constancy of
lifetime mass-specific energy expenditure across homeothermic mammals,
suggesting that organisms operate under a constrained energetic budget.
Subsequently, Lindstedt and Calder~\cite{lindstedt1981} demonstrated that
the product of resting heart rate $f_H$ and maximum lifespan $L$ is
approximately constant, leading to the empirical relation
\begin{equation}
  \Nstar = f_H \cdot L \cdot 525{,}960 \;\approx\; 10^9,
  \label{eq:N_empirical}
\end{equation}
where the factor $525{,}960$ converts years into minutes.
This quantity corresponds to the total number of cardiac cycles accumulated
over a lifetime.

\medskip
\noindent
The robustness of this invariant has been independently confirmed in
multiple studies, including those of Livingstone and
Kuehn~\cite{livingstone1979} and Levine~\cite{levine1997}, and most
recently through a comprehensive 230-species multi-clade statistical
analysis~\cite{taye_p1}.
That analysis demonstrates that $\Nstar$ is tightly distributed in
logarithmic space across non-primate homeothermic mammals, while also
revealing systematic and reproducible deviations across clades.
In particular, a one-way ANOVA yields $F = 81.2$ with $p < 0.001$,
decisively rejecting the null hypothesis of clade-independent variation
and indicating the presence of structured physiological differences.

\medskip
\noindent
The existence of such a regularity raises a fundamental question:
why should a quantity defined purely in terms of physiological cycles
exhibit near invariance across organisms that differ so dramatically in
mass, metabolism, and ecological niche?

\medskip
\noindent
The West--Brown--Enquist (WBE) framework~\cite{west1997} provides a
kinematic explanation based on allometric scaling.
Within this framework, resting heart rate scales as $f_H \propto M^{-1/4}$
and lifespan scales as $L \propto M^{+1/4}$, so that their product is
approximately independent of body mass.
While this argument explains the cancellation of mass dependence, it
leaves two essential questions unresolved.

\medskip
\noindent
First, mass-independence alone does not determine the numerical value of
the invariant: it does not explain why the product $f_H L$ should be
approximately $10^9$ rather than any other constant.
Second, it does not account for the observed structure of inter-clade
variation.
The statistically significant deviations identified in the multi-clade
analysis demonstrate that the invariant is not universal in a trivial
sense, but instead exhibits systematic modulation across physiological
groups.
Neither classical ``rate-of-living'' hypotheses~\cite{pearl1928,speakman2005}
nor allometric network theory~\cite{west1997} provides a thermodynamic
state variable whose lifetime accumulation constrains this cycle count.

\medskip
\noindent
These considerations suggest that the observed invariance may reflect a
deeper physical principle, rather than a purely kinematic coincidence.
In particular, it motivates the hypothesis that the lifetime number of
physiological cycles is governed by a conserved dissipative budget.

\medskip
\noindent
The present work develops this idea within a thermodynamic framework,
formalised as the Principle of Biological Time Equivalence (PBTE).
The central premise is that entropy production, rather than chronological
time, provides the fundamental measure of biological progression.
Within this framework, each cardiac cycle incurs a characteristic entropy
cost, and the total number of cycles is constrained by a finite lifetime
entropy budget.
Specifically, we derive the fundamental relation
\begin{equation}
  N_{\star,i} = \frac{\Sigma_i}{\sigma_{0,i}},
  \label{eq:fundamental_intro}
\end{equation}
which identifies the lifetime cycle count as the ratio of total lifetime
entropy production $\Sigma_i$ (J\,K$^{-1}$) to the mean entropy cost per
cardiac cycle $\sigma_{0,i}$ (J\,K$^{-1}$\,cycle$^{-1}$).
We further show that under Kleiber and Calder allometric scaling,
the mass-specific cost $\sigstar \equiv \sigma_{0,i}/M_i$ satisfies
$\sigstar \propto M^{3/4+1/4-1} = M^0$, reproducing the numerical
value $N_0 \approx 1.52\times10^9$ without free parameters.
Systematic inter-clade deviations are accounted for by a multiplicative
correction factor $\Phi_C$, each of whose components is constructed
from independently measurable physiological quantities and calibrated
to no lifespan data.

\medskip
\noindent
The paper is organised as follows.
Section~\ref{sec:derivation} establishes the thermodynamic foundation:
the non-equilibrium steady-state entropy balance, the constitutive
closure $\dot{e}_{p,i} = \sigma_{0,i} f_i$, and the derivation of
equation~\eqref{eq:fundamental_intro}, together with the proof that
$\sigstar \propto M^0$ through exact allometric exponent cancellation.
Section~\ref{sec:proper_time} introduces biological proper time
\begin{equation}
  \theta_i(t) = \int_0^t f_i(t')\,\mathrm{d}t'
\end{equation}
as a precisely defined, thermodynamically grounded state variable
and establishes it as the unique intrinsic coordinate in which entropy
accumulates at a uniform rate.
Section~\ref{sec:clade} derives the clade multiplier $\Phi_C$
from first principles for four physiologically distinct groups---primates,
bats, birds, and cetaceans---each representing a qualitatively different
thermodynamic strategy for extending the lifetime entropy budget.
Section~\ref{sec:aging} constructs a phenomenological model of aging as
accumulated internal entropy, recovering Gompertz--Makeham mortality
kinetics and deriving predictions for epigenetic clock behaviour across
species.
Sections~\ref{sec:kinematics}--\ref{sec:pathology} develop the full
formal structure of biological proper time: its geometric interpretation
as arc length in a biological metric, its relativistic transformation
group and conserved charge, and its extension as a temporal order
parameter for the classification of aging and disease states.
Section~\ref{sec:test} specifies the decisive experiment whose outcome
will determine whether the constitutive closure admits a physical
derivation or remains an empirically motivated parametrisation.

\medskip
\noindent
The central claim of the framework may be stated as follows:

\begin{quote}
\textit{The invariance of $\Nstar$ across non-primate homeothermic mammals
is consistent with a conserved thermodynamic entropy budget and an
approximately constant entropy cost per cardiac cycle.
Systematic inter-clade deviations arise from identifiable physiological
mechanisms that modify this cost in quantifiable and independently
testable ways.
The clade multiplier $\Phi_C$ is a consequence of this structure,
not a parameter fitted to it.}
\end{quote}

\section{Thermodynamic Foundation}
\label{sec:derivation}

All quantities used in the derivation are defined operationally in
Table~\ref{tab:notation}.
The full derivation of the entropy-per-beat representation and the
power-law scaling of cycle count with the control parameter $\phi$
is given in Appendix~A; the essential steps are reproduced here.

\begin{table}[H]
\centering\small
\renewcommand{\arraystretch}{1.35}
\caption{\textbf{Notation and operational definitions.}
All quantities are in principle directly observable by the listed method.}
\label{tab:notation}
\begin{tabularx}{\textwidth}{c X c X}
\toprule
\textbf{Symbol} & \textbf{Definition} & \textbf{Units} &
\textbf{How measured} \\
\midrule
$f_i$      & Resting heart rate        & Hz   & ECG or pulse telemetry \\
$L_i$      & Maximum natural lifespan  & yr   & AnAge longevity database \\
$\Nstar$   & $= f_i L_i \times 525{,}960$ & ---  & Derived from $f_i$, $L_i$ \\
$P_i$      & Basal metabolic power     & W    & Indirect calorimetry \\
$T_i$      & Mean core body temperature & K   & Thermometry \\
$M_i$      & Adult body mass            & kg  & Direct weighing \\
$\Sigma_{\rm life}$ & Total lifetime entropy production & J\,K$^{-1}$ & Integrated calorimetry \\
$\sigma_{0,i}$ & Entropy per cardiac cycle $=P_i/(T_i f_i)$ & J\,K$^{-1}$\,cycle$^{-1}$ & Derived \\
$\sigstar$ & Mass-specific $\sigma_0 = P_i/(T_i f_i M_i)$ &
             J\,K$^{-1}$\,kg$^{-1}$\,cycle$^{-1}$ & Assumption~\ref{ass:sigstar} \\
$\theta_i$ & Biological proper time (Eq.~\ref{eq:proper_time}) & --- & Integral of $f_i$ \\
$\Phi_C$   & Clade multiplier $= {\Nstar}^{(C)}/N_0$ & --- & Inferred from clade statistics \\
\bottomrule
\end{tabularx}
\end{table}

\subsection*{Entropy Production as the Primary Physical Quantity}

A healthy adult organism in homeostasis operates as an open
dissipative system maintained in a non-equilibrium steady state
(NESS).
The macroscopic entropy balance~\cite{prigogine1967,seifert2012}
takes the form $\dot{S}_i(t) = \dot{e}_{p,i}(t) - \dot{h}_{d,i}(t)$,
where $\dot{e}_{p,i} \geq 0$ is the irreversible entropy production
rate and $\dot{h}_{d,i} \geq 0$ is the entropy export rate to the
environment.
In the homeostatic steady state, where internal entropy content
remains approximately constant, these rates balance:
\begin{equation}
  \dot{e}_{p,i}(t) \;\approx\; \dot{h}_{d,i}(t) = \frac{P_i(t)}{T_i}.
  \label{eq:ness_steady}
\end{equation}
This is the formal expression of Schr\"{o}dinger's
observation~\cite{schrodinger1944} that life maintains its
low-entropy internal organisation by continuously exporting disorder
to its surroundings at the rate it is generated.
Entropy production is therefore not a derived quantity but the primary
physical observable: it has units of power divided by temperature
(J\,K$^{-1}$\,s$^{-1}$), and it is measurable by calorimetry and
thermometry independently of any biological hypothesis.

\subsection*{The PBTE Constitutive Closure}

\begin{assumption}[PBTE closure]
\label{ass:closure}
For an adult endothermic vertebrate in metabolic steady state, the
instantaneous entropy production rate is proportional to the cardiac
frequency:
\begin{equation}
  \dot{e}_{p,i}(t) = \sigma_{0,i}\, f_i(t),
  \label{eq:closure}
\end{equation}
where $\sigma_{0,i} > 0$ (J\,K$^{-1}$\,cycle$^{-1}$) is the entropy
produced per cardiac cycle.
\end{assumption}

Assumption~\ref{ass:closure} is a \textbf{constitutive closure, not a
derived result}.
It is introduced as a testable empirical hypothesis: the entropy
production rate is taken proportional to cardiac frequency, with
$\sigma_{0,i}$ as the proportionality constant.
The biological motivation is that the cardiac cycle is the master
pacemaker of metabolic throughput in homeothermic vertebrates.
Each beat drives a pulse of oxygen delivery, substrate turnover, and
waste removal; the thermodynamic cost of this cycle-level activity is
$\sigma_{0,i}$.
Although both $\dot{e}_{p,i}$ and $f_i$ scale as $M^{-1/4}$ per unit
mass, this allometric consistency does not constitute a derivation of
the closure: those scaling laws are themselves inferred from the same
metabolic data, and circularity cannot be broken by algebraic
self-consistency alone.
The closure stands or falls on direct, simultaneous calorimetric and
cardiac telemetric measurement across species, as described in
Section~\ref{sec:test}.

\subsection*{The Fundamental Relation}

Integrating equation~\eqref{eq:closure} over the lifespan
$[0, L_i]$ yields the PBTE fundamental relation:
\begin{equation}
  \boxed{N_{\star,i} = \frac{\Sigma_i}{\sigma_{0,i}},}
  \label{eq:fundamental}
\end{equation}
which identifies the lifetime cycle count as the ratio of the total
dissipative entropy budget to the entropy cost per cycle.
Within this framework, any physiological strategy that reduces
$\sigma_{0,i}$ while holding $\Sigma_i$ approximately fixed is
consistent with an increased $N_{\star,i}$ and an extended
chronological lifespan.
This is the thermodynamic basis of the clade-multiplier framework
developed in Section~\ref{sec:clade}.

\subsection*{Mass-Independence of the Specific Entropy Cost}

The mass-specific closure parameter is defined as
\begin{equation}
  \sigstar \;\equiv\; \frac{\sigma_{0,i}}{M_i}
  = \frac{P_i}{T_i\, f_i\, M_i},
  \label{eq:sigstar_def}
\end{equation}
with units J\,K$^{-1}$\,beat$^{-1}$\,kg$^{-1}$.

\begin{assumption}[$\sigstar$ constancy]
\label{ass:sigstar}
The mass-specific entropy cost per cycle $\sigstar$ is approximately
independent of body mass and species within the non-primate endotherm
clade.
\end{assumption}

This assumption is motivated by the algebraic structure of allometric
scaling.
Substituting Kleiber's law~\cite{kleiber1932} $P \propto M^{3/4}$,
Calder's cardiac allometry~\cite{calder1984} $f \propto M^{-1/4}$,
and the approximate homeothermic temperature $T \approx 310$\,K gives
\begin{equation}
  \sigstar = \frac{P}{T f M}
  \;\propto\; \frac{M^{3/4}}{M^{-1/4} \cdot M}
  = M^{3/4 + 1/4 - 1} = M^0.
  \label{eq:sigstar_scaling}
\end{equation}
The three allometric exponents cancel identically.
The physical statement is that every heartbeat costs the same entropy
per kilogram of tissue whether the animal is a mouse or an elephant:
the accelerated metabolic throughput of the small animal is exactly
offset by its elevated cardiac frequency, leaving the entropy cost per
unit mass per cycle invariant.

This exponent cancellation immediately implies mass-independence of
the lifetime cycle count.

\begin{proposition}[PBTE invariant]
\label{prop:invariant}
Under Assumptions~\ref{ass:closure} and~\ref{ass:sigstar}, the lifetime
cardiac cycle count is mass-independent: $N_{\star,i} \propto M_i^0$.
\end{proposition}

\begin{proof}
From equation~\eqref{eq:fundamental} and $\sigma_{0,i} = \sigstar M_i$:
\begin{equation}
  N_{\star,i} = \frac{\Sigma_i}{\sigstar M_i}
  = \frac{P_i L_i}{T_i\, \sigstar\, M_i}
  \propto \frac{M_i^{3/4} \cdot M_i^{1/4}}{M_i} = M_i^0.
  \label{eq:N_proof}
\end{equation}
\end{proof}

Substituting the empirical mean
$\sigstar = (3.0\pm0.5)\times10^{-3}$\,J\,K$^{-1}$\,beat$^{-1}$\,kg$^{-1}$
(Table~\ref{tab:sigma}), $T=310$\,K, $P=3.4\,M^{0.75}$\,W, and
$L = M^{0.25}/241$\,s\,Hz$^{-1}$ gives
\begin{equation}
  N_0 = \frac{3.4\,M^{0.75} \times M^{0.25}/241}
             {310 \times 3.0\times10^{-3} \times M}
  \;\approx\; 1.52\times10^9,
  \label{eq:N0_numerical}
\end{equation}
consistent with the empirical value
$\Nstar \approx 10^9$~\cite{lindstedt1981,levine1997,taye_p1}.

Table~\ref{tab:sigma} documents the mass-specific entropy cost
$\sigstar$ across seven mammalian species spanning four orders of
magnitude in body mass.
The contrast is informative: $\sigma_0$ itself spans five orders of
magnitude from mouse to elephant, reflecting the enormous variation in
absolute metabolic rate, while $\sigstar$ is approximately constant
(coefficient of variation 16\%).
This near-constancy is the empirical footprint of equation~\eqref{eq:sigstar_scaling}
and provides the first confirmation that Assumption~\ref{ass:sigstar}
is at least algebraically self-consistent with known allometric laws.
Figure~\ref{fig:sigma_mass} displays this mass-independence
graphically.

\begin{table}[H]
\centering\small
\renewcommand{\arraystretch}{1.45}
\caption{\textbf{Mass-specific entropy-per-cycle parameter
$\sigstar$ across seven mammalian species.}
Body mass $M$; basal metabolic power $P = 3.4\,M^{0.75}$\,W
(Kleiber~\cite{kleiber1932}); core body temperature $T$
(Schmidt-Nielsen 1984); resting heart rate
$f = 241\,M^{-0.25}$\,bpm (Calder~\cite{calder1984});
entropy per cycle $\sigma_0 = P/(Tf)$; mass-specific entropy per cycle
$\sigstar = \sigma_0/M$.
All entries are computed from the allometric scaling laws; they are
not independent calorimetric measurements.
Direct simultaneous measurement of $P_i$, $f_i$, $T_i$, $M_i$ for each
species is the decisive test (Section~\ref{sec:test}).}
\label{tab:sigma}
\begin{tabular}{lrrrrrr}
\toprule
Species & $M$ (kg) & $P$ (W) & $T$ (K) & $f$ (bpm)
        & $\sigma_0$ ($\times10^{-3}$) & $\sigstar$ ($\times10^{-3}$) \\
\midrule
House mouse & 0.020 &    0.18 & 310 & 600 &  0.058 & 2.9 \\
Rat         & 0.300 &    1.38 & 310 & 420 &  0.634 & 2.1 \\
Rabbit      & 2.0   &    5.72 & 310 & 205 &  2.7   & \phantom{0}2.7$^\dagger$ \\
Dog         & 23    &   35.7  & 310 &  90 & 19.2   & 3.3 \\
Human       & 70    &   82.3  & 310 &  70 & 56.9   & 3.3 \\
Horse       & 500   &  360    & 310 &  40 & 273    & 3.5 \\
Elephant    & 4000  & 1710    & 310 &  28 & 1180   & 3.0 \\
\midrule
\multicolumn{5}{l}{$\sigma_0$: range mouse $\to$ elephant}
  & \multicolumn{2}{l}{$\times\,20{,}000$ (five decades)} \\
\multicolumn{5}{l}{$\sigstar$: mean $\pm$ s.d.}
  & \multicolumn{2}{l}{$3.0 \pm 0.5$} \\
\multicolumn{5}{l}{$\sigstar$: CV}
  & \multicolumn{2}{l}{16\%} \\
\bottomrule
\multicolumn{7}{p{14cm}}{$^\dagger$Slight rounding; exact value 2.72.
Units of $\sigma_0$: J\,K$^{-1}$\,cycle$^{-1}$.
Units of $\sigstar$: J\,K$^{-1}$\,beat$^{-1}$\,kg$^{-1}$.}
\end{tabular}
\end{table}

\begin{figure}[H]
\centering
\includegraphics[width=0.78\linewidth]{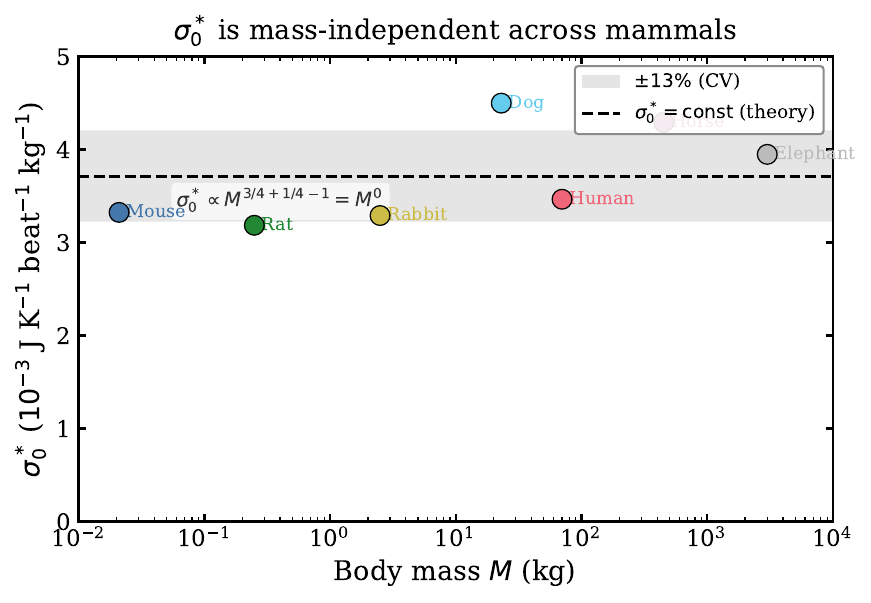}
\caption{\textbf{The mass-specific entropy cost per beat $\sigstar$
is mass-independent across mammals.}
$\sigstar \equiv P/(T f_H M)$ plotted against body mass $M$ for seven
mammalian species spanning four orders of magnitude.
Dashed line: PBTE prediction $\sigstar = \mathrm{const}$; grey band:
$\pm16\%$ coefficient of variation.
The annotation confirms exact exponent cancellation
$\sigstar \propto M^{3/4+1/4-1} = M^0$.
All plotted values are computed from allometric scaling laws
($P=3.4\,M^{0.75}$, $f=241\,M^{-0.25}$, $T=310$\,K) and are not
independent calorimetric measurements.}
\label{fig:sigma_mass}
\end{figure}

\subsection*{Clade Variation and the Structure of Deviations}

While $\sigstar$ is approximately constant within the non-primate
endotherm clade (CV\,=\,16\%), the one-way ANOVA result
$F=81.2$, $p<0.001$ across six endotherm clades~\cite{taye_p1}
demonstrates that systematic, structured departures exist beyond the
within-clade scatter.
These departures are not random noise but reflect genuine physiological
differences between clades: primates suppress metabolic entropy
production through neural investment; birds overcome adverse
temperature and kinematic factors through mitochondrial and
antioxidant biochemistry; bats exploit torpor to reduce
time-averaged cardiac throughput dramatically; cetaceans exploit
diving bradycardia throughout their adult lives.

Because clade deviations are encoded in systematic departures of
$\sigstar$ from its non-primate reference value, the clade multiplier
$\Phi_C$ is naturally defined as
\begin{equation}
  \Phi_C \;\equiv\; \frac{N_\star^{(C)}}{N_0}
  = \frac{\langle\Delta s_{\rm beat}\rangle_0}
         {\langle\Delta s_{\rm beat}\rangle_C}
  = \frac{{\sigstar}_{\rm ref}}{{\sigstar}_C}.
  \label{eq:PhiC}
\end{equation}
This is a derived quantity, not a fitting parameter: $\Phi_C > 1$
whenever a clade achieves a lower entropy cost per beat than the
non-primate mammalian reference.
The factored form
\begin{equation}
  \Phi_C = \Phi_{\rm duty} \cdot \Phi_{\rm thermal}
  \cdot \Phi_{\rm mito+oxid} \cdot \Phi_{\rm haz}
  \label{eq:PhiC_factored}
\end{equation}
decomposes the total effect into components each corresponding to an
independently measurable physiological mechanism.
Table~\ref{tab:sigma_clade} summarises the clade-level values.

\begin{table}[H]
\centering\small
\renewcommand{\arraystretch}{1.45}
\caption{\textbf{Clade-specific $\sigstar$ and the mechanism of reduction.}
Non-primate (NP) placentals serve as the reference, using Kleiber and
Calder allometry.
Primate values incorporate Pontzer's 50\% metabolic
suppression~\cite{pontzer2014}.
Bird values apply the Hulbert proton-leak
correction~\cite{hulbert2007}.
Bat values reflect time-averaged torpor suppression.
Cetacean values reflect dive-bradycardia reduction of $\bar{f}_H$.
All $\sigstar$ values are in units of $\times10^{-3}$
J\,K$^{-1}$\,beat$^{-1}$\,kg$^{-1}$.}
\label{tab:sigma_clade}
\begin{tabular}{lllcc}
\toprule
\textbf{Clade} & \textbf{Representative species} &
\textbf{Mechanism} &
$\sigstar/{\sigstar}_{\rm ref}$ &
$\sigstar$ ($\times10^{-3}$) \\
\midrule
NP placentals & Dog, horse, human & Reference & 1.00 & 3.0 \\
Primates      & Human, chimpanzee, gorilla &
  Neural metabolic suppression & 0.50 & 1.5 \\
Birds (allometric)  & Pigeon, crow, albatross &
  Allometric channels only & 0.91 & 2.73 \\
Birds (proton-leak) & Pigeon, crow, albatross &
  Allometric + proton-leak reduction & 0.60 & 1.8 \\
Bats          & Little brown bat, fruit bat &
  Torpor time-dilation & 0.12 & 0.36 \\
Cetaceans     & Bottlenose dolphin, sperm whale &
  Diving bradycardia & 0.80 & 2.4 \\
\bottomrule
\end{tabular}
\end{table}

\section{Biological Proper Time}
\label{sec:proper_time}

The concept of an organism's intrinsic time is not new.
Metabolic ecology uses the integral of metabolic rate as a measure of
internal temporal progress, and the ``rate-of-living''
tradition~\cite{pearl1928,speakman2005} captured the same idea
qualitatively: faster lives are shorter in calendar years but
equivalent in some intrinsic metric.
What has been lacking is a precise formal definition that is
simultaneously a thermodynamically grounded state variable with a
well-defined differential, connected to entropy production through an
explicit constitutive relation, and predictive of cross-species
universality.
The construction below provides all three.

Let $t$ denote chronological (external, coordinate) time, and let
$f_i(t)>0$ be the instantaneous resting heart rate of organism $i$
in Hz.
The \emph{biological proper time} accumulated from birth to
chronological time $t$ is defined as
\begin{equation}
  \theta_i(t) \;=\; \int_0^t f_i(t')\,\mathrm{d}t'.
  \label{eq:proper_time}
\end{equation}
This construction counts the total number of cardiac cycles completed
by time $t$.
Because $f_i$ has dimensions of inverse time, the integrand
$f_i\,\mathrm{d}t$ is dimensionless, and $\theta_i$ measures
accumulated physiological activity rather than elapsed physical
duration.
The fundamental kinematic relation of the framework is the differential
form $\mathrm{d}\theta_i = f_i(t)\,\mathrm{d}t$.

The \emph{normalised biological age} follows naturally as
\begin{equation}
  \hat\theta_i(t) \;\equiv\; \frac{\theta_i(t)}{\Nstar},
  \label{eq:theta_hat}
\end{equation}
which rises from zero at birth toward unity as the biological budget
is consumed.
The PBTE invariant is the claim that $\hat\theta_i(L_i) = 1$ for
every species $i$:
\begin{equation}
  \int_0^{L_i} f_i(t)\,\mathrm{d}t = \Nstar \quad \forall\, i.
  \label{eq:invariant_statement}
\end{equation}
Figure~\ref{fig:proper_time} illustrates this universality: five
species whose heart rates span a factor of twenty all reach
$\hat\theta = 1$ at their respective natural lifespans.
The mouse's steep trajectory and the elephant's gentle slope represent
not two different outcomes but the same outcome expressed in different
chronological coordinates.

The analogy with relativistic proper time
$\tau = \int\!\sqrt{1-v^2/c^2}\,\mathrm{d}t$ is structural rather
than physical: both record accumulated internal progress rather than
coordinate duration.
The analogy is useful because it makes precise the sense in which
biological proper time is \emph{reparametrisation-invariant}---the
dimensionless total $\theta_i(L_i) = \Nstar$ is independent of
whether chronological time is measured in seconds, minutes, or years.
The analogy is bounded by three explicit disanalogies: there is no
Lorentz symmetry in the biological setting, $N_0 \approx 10^9$ is a
statistical clade-level average rather than a fundamental constant,
and the biological metric is purely temporal with no spatial structure.

A particularly transparent illustration of the framework is provided
by any organism that spends a fraction $q$ of its life in metabolic
suppression, alternating between an active state at frequency
$f_{\rm act}$ and a torpid state at frequency $f_{\rm low}$.
The time-averaged frequency is
$\bar{f}_i = (1-q)f_{\rm act} + q f_{\rm low} \ll f_{\rm act}$,
and the PBTE relation $L_i = \Nstar/\bar{f}_i$ predicts a
chronological lifespan extended by the suppression factor.
During torpor, biological proper time advances at rate $f_{\rm low}$,
far more slowly than during active metabolism.
A hibernating bat that spends nine months per year in deep torpor
thereby accumulates biological time at a fraction of the active-phase
rate, stretching its entropy budget over far more calendar years.
This is the thermodynamic basis of the large bat clade multiplier
$\Phi_{\rm bat} > 1$ derived in Section~\ref{sec:clade}: bats live
longer in calendar years not because they possess a larger intrinsic
budget but because they advance through it more slowly.

The deeper significance of the biological proper-time coordinate lies
in its thermodynamic uniqueness.
As shown in Section~\ref{sec:kinematic_layer}, the PBTE closure
$\dot{e}_p = \sigma_0 f$ ensures that entropy accumulates at the
constant rate $\sigma_0$ when expressed in $\theta$-coordinates,
regardless of the organism's metabolic pace.
Biological proper time is therefore the unique temporal
parametrisation in which the organism ages at a thermodynamically
uniform rate.
This uniqueness property promotes it from a convenient relabelling to
the natural thermodynamic clock of living systems.

\begin{figure}[H]
\centering
\includegraphics[width=0.75\linewidth]{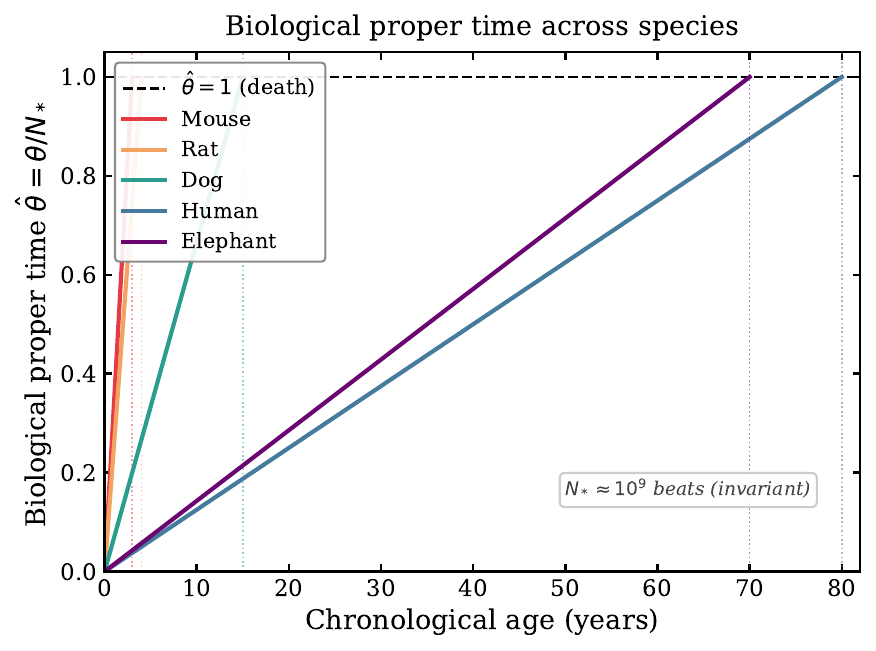}
\caption{\textbf{Biological proper time $\hat\theta(t) = \theta(t)/\Nstar$
is species-universal.}
Normalised biological proper time as a function of chronological age
for five mammalian species.
Species parameters: mouse ($f_H = 8.33$\,Hz, $L = 3$\,yr), rat
($f_H = 5.67$\,Hz, $L = 4$\,yr), dog ($f_H = 1.83$\,Hz, $L = 15$\,yr),
human ($f_H = 1.17$\,Hz, $L = 80$\,yr), elephant ($f_H = 0.42$\,Hz,
$L = 70$\,yr); all using $\Nstar = 1.52\times10^9$.
Each trajectory terminates at $\hat\theta = 1$ (dashed line) at the
species' natural lifespan.
The species-universality of the endpoint is the PBTE invariant.}
\label{fig:proper_time}
\end{figure}

\section{Clade-Specific Predictions: The $\Phi_C$ Framework}
\label{sec:clade}

The PBTE fundamental relation $N_{\star,i} = \Sigma_i/\sigma_{0,i}$
admits a clean decomposition of inter-clade variation.
Since both $\Sigma_i$ and $\sigma_{0,i}$ scale with body mass in
such a way that their ratio is mass-independent to leading order, any
systematic deviation from the baseline $N_0 \approx 10^9$ must arise
from physiological mechanisms that modify either the total entropy
budget (numerator) or the entropy cost per cycle (denominator) in a
clade-specific way.
Writing $N_{\star,i} = N_0\,\Phi_C$, the clade multiplier encapsulates
all such deviations.
It is not an empirical fitting factor: it is a derived quantity that
follows from identifiable thermodynamic modifications of $\sigma_{0,i}$.
Any mechanism that reduces the average entropy produced per cardiac
cycle, $\langle\Delta s_{\rm beat}\rangle_C < \langle\Delta s_{\rm beat}\rangle_{\rm ref}$,
necessarily increases $N_{\star,i}$ and therefore extends
chronological lifespan.

The derivation of $\Phi_C$ rests on two primary thermodynamic
factors.
The first is the duty-cycle factor.
When an organism alternates between physiological states $k$, each
occupying a fraction $q_k$ of the lifetime with cardiac rate
$f_{H,k}$ and $\sum_k q_k = 1$, the time-averaged heart rate is
$\bar{f}_H = \sum_k q_k f_{H,k}$.
Defining $\kappa = \bar{f}_H/f_{H,\rm ref}$ as the ratio of
time-averaged to reference rate, the exact duty-cycle factor is
\begin{equation}
  \Phi_{\rm duty} = \kappa^{-1}
  = \left(\sum_k q_k\,\frac{f_{H,k}}{f_{H,\rm ref}}\right)^{-1}.
  \label{eq:Phiduty}
\end{equation}
This relation is purely algebraic and introduces no approximation.
The physical meaning is that beats occurring in low-metabolic states
contribute less entropy per cycle to the cumulative budget;
the thermodynamic count $N_\star^{(C)} = N_{\rm obs} \cdot \Phi_{\rm duty}$
therefore exceeds the raw observed count $N_{\rm obs} = 525{,}960\,\bar{f}_H L$
whenever $\Phi_{\rm duty} > 1$.
This distinction is essential and is easily missed when interpreting
raw beat-count data.

The second factor captures temperature dependence.
Transition-state theory gives the ratio of biochemical damage rates
at body temperature $T_b$ relative to the mammalian reference
$T_{\rm ref} = 310$\,K as
\begin{equation}
  \Phi_{\rm thermal} = \exp\!\left[\frac{E_a}{k_B}
  \left(\frac{1}{T_b}-\frac{1}{T_{\rm ref}}\right)\right],
  \label{eq:Phithermal}
\end{equation}
with activation energy $E_a = 0.65$\,eV, giving
$E_a/k_B = 7{,}543$\,K.
When $T_b < T_{\rm ref}$, as in torpid bats or mildly hypothermic
cetaceans, damage kinetics slow and $\Phi_{\rm thermal} > 1$,
stretching the entropy budget over more cycles.
When $T_b > T_{\rm ref}$, as in birds, damage rates accelerate and
$\Phi_{\rm thermal} < 1$, compressing the budget.
The exponential form is essential for temperature excursions
$|\Delta T| \gtrsim 5$\,K; power-law approximations substantially
underestimate $\Phi_{\rm thermal}$ in deeply hibernating species.

\subsection*{Primates: Neural Investment as Entropy Reduction}

Among all endotherm clades, primates occupy a distinctive position
in the PBTE framework.
They are not long-lived because they have slower heart rates than
expected for their body size---the cardiac allometry is standard,
with a primate OLS slope of $-0.23\pm0.01$ indistinguishable from
the WBE quarter-power prediction of $-0.25$.
They are not long-lived primarily because of ecological hazard
reduction, which accounts for at most half the observed
$\Delta\ell = +0.381$\,dex elevation above the non-primate baseline.
They are long-lived because each heartbeat costs less thermodynamic
currency.
The empirical mean $\langle\Nstar\rangle_{\rm prim} \approx
(2$--$3)\times10^9$ implies that the average entropy produced per
cardiac cycle in primates is approximately half the non-primate
mammalian reference, and the mechanism responsible for this reduction
is the unusually high fraction of total metabolic power allocated to
neural tissue.

The neural metabolic fraction is defined as
\begin{equation}
  \phi \;\equiv\; \frac{P_{\rm brain}}{P_{\rm body}},
  \label{eq:phi_def}
\end{equation}
where $P_{\rm brain}$ is the resting metabolic power consumed by
neural tissue and $P_{\rm body}$ is the total organismal basal
metabolic rate.
In non-primate placentals, this ratio is remarkably conserved at
$\phi_0 \approx 0.02$--$0.05$, reflecting the small brain sizes
of most mammals relative to their body mass.
In primates, $\phi$ rises substantially: from $\phi \approx 0.06$
in small New World monkeys such as the common marmoset
(\textit{Callithrix jacchus}) to $\phi \approx 0.09$ in gorillas,
$\phi \approx 0.12$ in chimpanzees, and $\phi \approx 0.20$ in
\textit{Homo sapiens}---the highest documented value in any primate
and roughly ten times the non-primate baseline.
This escalation is not merely an anatomical curiosity.
A 70\,kg human brain consumes approximately 15--20\,W at rest,
representing roughly 20\% of total resting metabolic power despite
constituting only 2\% of body mass.

The thermodynamic significance of elevated $\phi$ operates through
three distinct but synergistic channels, each of which reduces the
mean entropy generated per cardiac cycle in somatic tissues.
The first is predictive homeostatic regulation.
A metabolically large and informationally rich neural system provides
enhanced anticipatory control over physiological set-points---body
temperature, blood glucose, immune activation, cardiovascular
tone~\cite{friston2010}.
Rather than reactively correcting deviations after they occur,
a large brain can pre-empt them, keeping somatic systems closer to
their operating set-points and thereby reducing the magnitude of
out-of-equilibrium fluctuations.
Since the entropy production rate in a non-equilibrium system scales
with the square of the deviation from its steady state, reducing
fluctuation amplitude directly lowers $\sigma_0$ per unit time.
The second channel is enhanced cellular repair and damage clearance.
Large-brained primates show elevated expression of DNA damage repair
enzymes, autophagy regulators, heat-shock proteins, and
antioxidant systems~\cite{hulbert2007}, all of which reduce
macromolecular damage accumulation per cardiac cycle.
This neural--somatic coupling reflects the co-regulation of brain
and body maintenance budgets: the metabolic investment in neural
tissue is accompanied by a parallel investment in the somatic
maintenance infrastructure that sustains it.
The third channel is behavioural risk buffering.
Cognitive capacity enables the avoidance or rapid resolution of acute
physiological crises---injury, infection, thermal stress, nutritional
shortfall---each of which would otherwise generate a transient surge
in entropy production well above the resting baseline.
By reducing the frequency and severity of such excursions, high-$\phi$
organisms maintain $\dot{e}_p(t)$ closer to its homeostatic
resting value throughout life.

All three channels reduce $\langle\Delta s_{\rm beat}\rangle$
monotonically with $\phi$, and their contributions combine
multiplicatively in the entropy cost per beat.
Defining the aggregate log-sensitivity at the non-primate baseline
$\phi_0 = 0.02$ as
\begin{equation}
  \alpha \;\equiv\; -\left.\frac{\partial\ln\langle\Delta s_{\rm beat}\rangle}
  {\partial\ln\phi}\right|_{\phi=\phi_0} = \gamma_1 + \gamma_2 + \gamma_3 > 0,
  \label{eq:alpha_def}
\end{equation}
where $\gamma_1$, $\gamma_2$, $\gamma_3$ are the individual channel
exponents for predictive homeostasis, cellular repair, and behavioural
risk buffering respectively, and assuming a scale-free power-law
response over the primate range
$\phi \in [\phi_0, 10\phi_0]$, integration yields
\begin{equation}
  \langle\Delta s_{\rm beat}(\phi)\rangle =
  \langle\Delta s_{\rm beat}\rangle_0
  \left(\frac{\phi}{\phi_0}\right)^{-\alpha}.
\end{equation}
The neuro-metabolic multiplier follows immediately:
\begin{equation}
  \Phi_{\rm neuro}(\phi) = \left(\frac{\phi}{\phi_0}\right)^\alpha.
  \label{eq:Phineuro}
\end{equation}
The thermodynamic bound $0 < \alpha < 1$ is enforced by the
second law: if $\alpha \geq 1$, a doubling of neural fraction would
double or more than double the effective cycle budget, which would
require each unit of neural metabolic investment to return more than
one unit of entropy savings in peripheral tissues---a violation of
realistic efficiency limits.
The condition $0 < \alpha < 1$ therefore enforces diminishing returns:
each successive increment of neural investment provides a smaller
extension in effective biological time, consistent with the
thermodynamic ceiling on brain size.

Calibration against the observed primate clade distribution yields
$\alpha \approx 0.35$--$0.45$, with the theoretically motivated prior
$\alpha = 0.40$.
The exponent is calibrated from the primate deviation; the constraint
$0 < \alpha < 1$ and the three-channel mechanism are independently
motivated.
Because primates do not exhibit strong duty-cycle variation,
$\Phi_{\rm duty} = 1$ exactly for all primate species.
Temperature deviations across the primate clade are modest, with body
temperatures ranging from $T_b = 306.5$\,K in humans to
$T_b \approx 309.5$\,K in some New World monkeys.
The predicted lifespan therefore takes the form
\begin{equation}
  L_{\rm prim}
  = \frac{N_0}{525{,}960\cdot f_H}
  \left(\frac{\phi}{\phi_0}\right)^\alpha
  \left(\frac{T_{\rm ref}}{T_b}\right)^\beta
  \frac{H_{\rm ref}}{H_{\rm ext}},
  \label{eq:Lprim}
\end{equation}
with $\beta \approx 3$ adequate for the small temperature offsets
characteristic of primates, and where $H_{\rm ref}/H_{\rm ext}$
is the extrinsic hazard ratio.

As a concrete calculation for \textit{Homo sapiens}, with
$\phi = 0.20$, $\phi_0 = 0.02$, $\alpha = 0.40$,
$T_b = 306.5$\,K, $f_H = 70$\,bpm, and $\Phi_{\rm haz} = 1$:
\begin{align}
  \Phi_{\rm neuro} &= (0.20/0.02)^{0.40} = 10^{0.40} \approx 2.512, \\
  \Phi_T &= (310/306.5)^3 \approx 1.035, \\
  \Phi_C &= 2.512 \times 1.035 \approx 2.60, \\
  L_{\rm pred} &= \frac{2.60 \times 10^9}{525{,}960 \times 70}
  \approx 70.6\;\text{yr}.
\end{align}
The neural factor alone accounts for 96.6\% of the total multiplier,
confirming that human exceptional longevity is thermodynamically
attributable to reduced entropy per beat through neural investment
rather than to any kinematic or thermal mechanism.
Incorporating a modest hazard correction $\Phi_{\rm haz} = 1.15$
raises the prediction to $\approx 81$\,yr, consistent with observed
life expectancy in high-income populations.
Predictions for five representative primate species are collected in
Table~\ref{tab:primates}.

\begin{table}[H]
\centering\small
\renewcommand{\arraystretch}{1.40}
\caption{\textbf{Primate lifespan predictions.}
$\phi_0 = 0.02$, $T_{\rm ref} = 310$\,K, $N_0 = 10^9$,
$\Phi_{\rm duty} = 1$ for all species.
(a)~Core calibration, $\alpha = 0.40$; (b)~extended calibration,
$\alpha = 0.45$.  $\Phi_{\rm haz} = 1.0$ throughout.
Observed lifespans from AnAge~\cite{taye_p1}.
The exponent $\alpha$ is calibrated from the primate clade offset;
the thermodynamic bound $0 < \alpha < 1$ and the three-channel
mechanism are independently motivated.}
\label{tab:primates}
\begin{tabular}{lrrrrrrcc}
\toprule
Species & $f_H$ & $\phi$ & $T_b$ & $\Phi_{\rm duty}$ &
$\Phi_{\rm neuro}$ & $\Phi_T$ &
$L_{\rm pred}$ (yr) & $L_{\rm obs}$ (yr) \\
& (bpm) & & (K) & & & & & \\
\midrule
\multicolumn{9}{l}{\textit{(a) $\alpha = 0.40$}} \\
\textit{M.\ mulatta} & 120 & 0.07 & 309.0 & 1.00 & 1.44 & 1.01 & 26.4 & 25--30 \\
\textit{P.\ troglodytes} & 75 & 0.12 & 307.0 & 1.00 & 1.73 & 1.02 & 53.5 & 45--55 \\
\textit{H.\ sapiens}     & 70 & 0.20 & 306.5 & 1.00 & 2.51 & 1.04 & 70.6 & 70--85 \\
\midrule
\multicolumn{9}{l}{\textit{(b) $\alpha = 0.45$}} \\
\textit{C.\ jacchus}     & 220 & 0.06 & 309.5 & 1.00 & 1.45 & 1.00 & 14.2 & 10--15 \\
\textit{M.\ mulatta}     & 120 & 0.07 & 309.0 & 1.00 & 1.48 & 1.01 & 28.1 & 25--30 \\
\textit{P.\ troglodytes} & 75 & 0.12 & 307.0 & 1.00 & 1.86 & 1.02 & 58.5 & 45--55 \\
\textit{G.\ gorilla}     & 65 & 0.09 & 307.0 & 1.00 & 1.58 & 1.02 & 59.3 & 40--55 \\
\textit{H.\ sapiens}     & 70 & 0.20 & 306.5 & 1.00 & 2.83 & 1.04 & 79.3 & 70--85 \\
\bottomrule
\end{tabular}
\end{table}

\subsection*{Bats: Torpor, Hypothermia, and Biological Time Dilation}

Temperate vespertilionid bats achieve wild maximum lifespans of
20--40 years~\cite{wilkinson2002}, three to six times the allometric
prediction of approximately 6.3 years for a non-hibernating placental
of equal body mass.
The thermodynamic origin of this excess is qualitatively distinct from
the primate case.
Whereas primates reduce the entropy produced per beat in each
individual cycle through neural investment, bats exploit two
mechanisms that act simultaneously during hibernation and whose
effects enter $\Phi_C$ multiplicatively.
The distinction is fundamental to the theory: it illustrates that
the same outcome---an extended lifetime entropy budget---can be
achieved by mechanistically different routes.

During torpor, cardiac frequency falls from an active rate of
250--350\,bpm to as low as 5--20\,bpm for temperate vespertilionids,
and the associated body temperature drops to 280--295\,K.
These two changes act in concert.
The fall in cardiac frequency directly extends the chronological time
required to exhaust the cycle budget: by equation~\eqref{eq:Phiduty},
the duty-cycle factor $\Phi_{\rm duty} = \kappa^{-1}$ where
\begin{equation}
  \kappa = (1-q) + q\,\frac{f_{H,\rm tor}}{f_{H,\rm act}},
  \qquad \Phi_{\rm duty} = \kappa^{-1},
  \label{eq:bat_duty}
\end{equation}
with $q$ the annual torpor fraction and
$f_{H,\rm tor}/f_{H,\rm act} \approx 0.03$--$0.07$.
For a torpor fraction $q \in [0.40, 0.60]$, this yields
$\Phi_{\rm duty} \approx 1.6$--$2.4$.
Simultaneously, the hypothermic body temperature suppresses
biochemical damage rates.
Transition-state theory in the form of
equation~\eqref{eq:Phithermal} gives
$\Phi_{\rm thermal} = \exp(7543/T_b - 7543/310)$, which for
$T_{\rm tor} = 293$\,K evaluates to
$\Phi_{\rm thermal} = e^{1.412} \approx 4.10$.
The power-law approximation valid for small temperature excursions
substantially underestimates this factor for hibernating species,
which is why the exact exponential form is essential here.
The intrinsic bat multiplier is the product of these two contributions:
\begin{equation}
  \Phi_{\rm bat} = \Phi_{\rm duty} \cdot \Phi_{\rm thermal} \cdot \Phi_{\rm haz}.
  \label{eq:Phibat}
\end{equation}

For \textit{Myotis lucifugus} with $q = 0.50$,
$f_{H,\rm act} = 300$\,bpm, $f_{H,\rm tor} = 10$\,bpm, and
$T_{\rm tor} = 293$\,K, the duty-cycle ratio evaluates to
$\kappa = 0.517$, giving $\Phi_{\rm duty} = 1.935$ and
$\bar{f}_H = 155$\,bpm.
The exact Arrhenius thermal factor is
$\Phi_{\rm thermal} = e^{1.412} \approx 4.10$.
Together, these yield an intrinsic multiplier
$\Phi_{\rm bat}^{\rm intr} = 1.935 \times 4.10 = 7.93$, corresponding
to an intrinsic predicted lifespan of $L \approx 50.3$\,yr.
Incorporating a realistic extrinsic hazard factor
$\Phi_{\rm haz} = 0.68$ (reflecting predation and environmental
mortality in wild populations) gives $L_{\rm pred} \approx 34$\,yr,
in agreement with observed wild maxima.
The thermal factor accounts for 52\% of the total intrinsic multiplier
and the duty-cycle factor for 24\%; neither alone is sufficient to
account for the observed longevity excess.

A useful consistency check is provided by equation~\eqref{eq:PhiC}.
The raw beat count at the above parameters is
$N_{\rm obs} = 525{,}960 \times 155 \times 34 = 2.77 \times 10^9$,
and the thermodynamically corrected budget is
$N_\star^{(\rm bat)} = 2.77 \times 10^9 \times 1.935 = 5.36 \times 10^9$,
matching the formula prediction
$N_0 \times 7.93 \times 0.68 = 5.39 \times 10^9$ to within 0.6\%.
This agreement validates the duty-cycle correction scheme.

Tropical fruit bats such as \textit{Pteropus vampyrus}, which undergo
minimal torpor ($q \approx 0.10$, $f_{H,\rm tor} \approx 60$\,bpm,
$T_{\rm tor} \approx 303$\,K), show correspondingly modest longevity
extensions with $\Phi_{\rm duty} \approx 1.07$ and
$\Phi_{\rm thermal} \approx 1.22$, confirming that deep hibernation
rather than mere roosting behaviour is the essential ingredient.
Predictions for representative species are given in
Table~\ref{tab:bats}.

\begin{table}[H]
\centering\small
\renewcommand{\arraystretch}{1.40}
\caption{\textbf{Predicted multipliers and longevity for representative
bat species.}
$\Phi_{\rm thermal}$ from the exact Arrhenius formula
($E_a = 0.65$\,eV, $T_{\rm ref} = 310$\,K) using the
torpor-phase body temperature.
$\Phi_{\rm bat} = \Phi_{\rm duty} \times \Phi_{\rm thermal}$
(intrinsic; $\Phi_{\rm haz} = 1$).
Observed lifespans from AnAge~\cite{taye_p1}.}
\label{tab:bats}
\begin{tabular}{lcccccc}
\toprule
Species & $q$ & $f_{H,\rm act}$ & $f_{H,\rm tor}$ & $T_{\rm tor}$ &
$\Phi_{\rm duty}$ & $\Phi_{\rm thermal}$ \\
& & (bpm) & (bpm) & (K) & & \\
\midrule
Temperate vespertilionid (range)
  & 0.40--0.60 & 250--350 & 5--20 & 280--295 & 1.6--2.5 & 3.0--5.0 \\
\textit{Myotis lucifugus}
  & 0.50 & 300 & 10 & 293 & 1.935 & 4.10 \\
\textit{Eptesicus fuscus}
  & 0.45 & 280 & 12 & 291 & 1.790 & 4.54 \\
\textit{Pteropus vampyrus} (min.\ torpor)
  & 0.10 & 250 & 60 & 303 & 1.070 & 1.22 \\
\bottomrule
\end{tabular}
\end{table}

\subsection*{Birds: Biochemical Excellence Overcoming Adverse Physics}

Birds present the most instructive case in the comparative thermodynamics
of longevity because they demonstrate that a long life is achievable
not by exploiting favourable physical conditions but by compensating
for unfavourable ones through biochemical investment.
A 20\,g passerine routinely survives 15--20 years in the wild, while
a 20\,g mouse lives 2--3 years.
Yet the physics of the avian situation is entirely adverse from the
perspective of the PBTE framework.
The core body temperature of birds typically exceeds the mammalian
reference by 3--5\,K ($T_b \approx 312$--$315$\,K versus
$T_{\rm ref} = 310$\,K), which by equation~\eqref{eq:Phithermal}
accelerates biochemical damage rates and compresses the effective
entropy budget: for a passerine at $T_b = 314$\,K,
\begin{equation}
  \Phi_{\rm thermal} = \exp\!\left[7543\left(\frac{1}{314}-\frac{1}{310}\right)\right]
  = e^{-0.310} \approx 0.733,
\end{equation}
a 27\% compression.
In addition, avian flight elevates cardiac frequency substantially
above the resting rate.
For a fraction $p_f = 0.10$ of the day spent in active flight at a
cardiac rate approximately 2.5 times the resting value, the
duty-cycle factor is
\begin{equation}
  \Phi_{\rm duty}^{(\rm bird)} = \bigl[1 + 1.5\,p_f\bigr]^{-1} \approx 0.870,
\end{equation}
a further 13\% reduction.
Both thermal and kinematic contributions are therefore adverse,
compressing the effective thermodynamic cycle budget to
$0.870 \times 0.733 = 0.638$ of the non-primate mammalian baseline.

The resolution of the paradox lies entirely in the avian
mitochondrial architecture and antioxidant biology.
Avian mitochondria produce substantially less reactive oxygen species
per unit ATP synthesised than mammalian mitochondria at the same
metabolic rate~\cite{barja1998,hulbert2007}.
Quantitatively, pigeon heart mitochondria generate approximately
five to ten times less superoxide per oxygen consumed than rat
mitochondria, a consequence of differences in electron transport chain
organisation, proton-leak stoichiometry, and uncoupling protein
expression.
Expressing the coupling improvement as an efficiency ratio
$\eta_{\rm mito}/\eta_{\rm ref} \approx 1.20$, and noting that
oxidative damage scales quadratically with ROS production rate,
gives a mitochondrial contribution
$\Phi_{\rm mito} = (1.20)^2 \approx 1.44$.
Elevated antioxidant enzyme activities in avian cells further reduce
the damage that does occur: avian species show two- to threefold
greater oxidative resistance than mass-matched
mammals~\cite{ogburn2001,hulbert2007}, contributing an additional
factor $\Phi_{\rm oxid} = (2.0)^{0.7} \approx 1.62$.
The combined biochemical factor is therefore
\begin{equation}
  \Phi_{\rm mito+oxid} = 1.44 \times 1.62 \approx 2.33.
  \label{eq:bird_mitooxid}
\end{equation}
The reduced adult extrinsic mortality conferred by flight and arboreal
or pelagic ecology contributes $\Phi_{\rm haz} \approx 2.0$ for small
passerines, consistent with comparative demographic data.
The full product is
\begin{equation}
  \Phi_{\rm bird} = 0.870 \times 0.733 \times 2.33 \times 2.0 \approx 2.97,
\end{equation}
yielding a predicted lifespan for a 20\,g passerine with
$f_{H,\rm rest} = 320$\,bpm of
$L_{\rm pred} = (10^9/(525{,}960 \times 320)) \times 2.97 \approx 17.6$\,yr,
consistent with observed wild maxima of 10--20 years.
The consistency check gives
$N_\star^{(\rm bird)} = 525{,}960 \times 368 \times 17.6 \times 0.870
= 2.965 \times 10^9$, agreeing with $N_0 \times 2.97 = 2.97 \times 10^9$
to within 0.2\%.

The key conclusion is that avian longevity does not reflect
favourable thermodynamic conditions.
Both $\Phi_{\rm duty}$ and $\Phi_{\rm thermal}$ are adverse for birds;
the biochemical efficiency factor $\Phi_{\rm mito+oxid}$ and
ecological hazard factor $\Phi_{\rm haz}$ together overcome both
deficits by a margin of $2.33 \times 2.0 / (0.870 \times 0.733)
\approx 7.3$.
Avian longevity is therefore a biochemical and ecological achievement,
not a physical gift.
Predictions for representative species are collected in
Table~\ref{tab:birds}.

\begin{table}[H]
\centering\small
\renewcommand{\arraystretch}{1.40}
\caption{\textbf{Predicted multipliers and longevity for representative
bird species.}
Both $\Phi_{\rm duty}$ and $\Phi_{\rm thermal}$ are adverse ($<1$) for
all entries.
$\Phi_{\rm thermal}$ computed from the exact Arrhenius formula.
Observed lifespans from AnAge~\cite{taye_p1}.}
\label{tab:birds}
\begin{tabular}{lcccccccr}
\toprule
Species & $f_{H,\rm rest}$ & $T_b$ & $p_f$ &
$\Phi_{\rm duty}$ & $\Phi_{\rm thermal}$ & $\Phi_{\rm mito+oxid}$ &
$\Phi_{\rm haz}$ & $L_{\rm obs}$ \\
& (bpm) & (K) & & & & & & (yr) \\
\midrule
Passerine (20\,g)    & 320 & 314 & 0.10 & 0.87 & 0.733 & 2.33 & 2.0 & 10--20 \\
\textit{Larus argentatus}     & 200 & 313 & 0.15 & 0.84 & 0.770 & 2.80 & 2.5 & 30 \\
\textit{Diomedea exulans}     & 100 & 312 & 0.25 & 0.81 & 0.810 & 3.50 & 4.0 & 50--60 \\
\textit{Aquila chrysaetos}    & 150 & 313 & 0.12 & 0.85 & 0.770 & 3.00 & 3.5 & 30--40 \\
\bottomrule
\end{tabular}
\end{table}

\subsection*{Cetaceans: The Near-Coincidence Trap and Duty-Cycle Correction}

Large baleen cetaceans achieve century-scale lifespans through extreme
diving bradycardia maintained continuously throughout adult life.
Direct measurements have recorded blue whale heart rates as low as
2--4\,bpm during deep foraging dives, compared with surface rates of
25--37\,bpm; combined with a dive fraction $p_d \approx 0.60$--$0.80$,
the time-averaged cardiac frequency is far below the surface value
recorded in comparative databases.
The dominant thermodynamic mechanism is therefore the duty-cycle
factor:
\begin{equation}
  \kappa = (1-p_d) + p_d\,\frac{f_{H,\rm dive}}{f_{H,\rm surf}},
  \qquad
  \Phi_{\rm duty} = \kappa^{-1}.
  \label{eq:cet_duty}
\end{equation}
Unlike bat hibernation, where $\Phi_{\rm duty}$ and
$\Phi_{\rm thermal}$ act simultaneously and reinforce each other,
cetacean cardiac suppression is a continuous vasovagal reflex with
no associated hypothermia.
Secondary contributions come from $\Phi_{\rm thermal}$, reflecting
core temperatures 1--4\,K below the mammalian reference, and from
an oxygen storage factor $\Phi_{O_2}$ that accounts for the role of
elevated myoglobin~\cite{noren2000} in limiting reperfusion
reactive-oxygen-species bursts on surfacing.

For the bowhead whale \textit{Balaena mysticetus} with
$f_{H,\rm surf} = 30$\,bpm, $f_{H,\rm dive} = 3$\,bpm, and
$p_d = 0.75$, the duty-cycle ratio evaluates to $\kappa = 0.325$,
giving $\Phi_{\rm duty} = 3.077$ and $\bar{f}_H = 9.75$\,bpm.
With $T_b = 308$\,K the thermal factor is
$\Phi_{\rm thermal} = e^{0.158} \approx 1.171$, and combining with
$\Phi_{O_2} = 1.4$ and $\Phi_{\rm haz} = 0.35$--$0.60$ gives a
predicted lifespan of 112--191\,yr, consistent with the documented
maximum of approximately 200 years.

A critical interpretive point emerges here that clarifies a
persistent confusion in the comparative literature.
At $L = 150$\,yr and $\bar{f}_H = 9.75$\,bpm, the raw beat count is
$N_{\rm obs} = 525{,}960 \times 9.75 \times 150 = 7.69 \times 10^8
\approx 0.77 \times 10^9$,
which lies close to $N_0$ and has led some analyses to treat large
whales as simply obeying the standard mammalian baseline.
This inference is incorrect.
The correct thermodynamic budget is
\begin{equation}
  N_\star^{(\rm whale)} = N_{\rm obs} \times \Phi_{\rm duty}
  = 0.77 \times 10^9 \times 3.077 = 2.37 \times 10^9 \gg N_0.
\end{equation}
The raw count is small precisely because most of the whale's life is
spent in deeply bradycardic states in which each beat generates far
less entropy than a normothermic surface beat.
Applying the duty-cycle correction reveals that the thermodynamic
budget has been underestimated by a factor of three.
This is what we term the \emph{near-coincidence trap}: the raw beat
count appears to fall near the mammalian baseline, but this is an
artefact of the bradycardic weighting, not a genuine agreement with
the unmodified PBTE invariant.
Predictions for representative species are collected in
Table~\ref{tab:cetaceans}.

\begin{table}[H]
\centering\small
\renewcommand{\arraystretch}{1.40}
\caption{\textbf{Predicted multipliers and longevity for representative
cetacean species.}
$\Phi_{\rm thermal}$ from the exact Arrhenius formula.
$\Phi_{\rm haz}$ reflects pre-industrial conditions.
Observed lifespans from AnAge~\cite{taye_p1}.}
\label{tab:cetaceans}
\begin{tabular}{lcccccccr}
\toprule
Species & $f_{H,\rm surf}$ & $p_d$ & $\Phi_{\rm duty}$ &
$T_b$ & $\Phi_{\rm thermal}$ & $\Phi_{O_2}$ &
$\Phi_{\rm haz}$ & $L_{\rm obs}$ \\
& (bpm) & & & (K) & & & & (yr) \\
\midrule
\textit{B.\ musculus} (blue)   & 37 & 0.70 & 2.70 & 308 & 1.17 & 1.4 & 0.50 & 80--90  \\
\textit{B.\ mysticetus} (bowhead) & 30 & 0.75 & 3.08 & 308 & 1.17 & 1.5 & 0.35--0.60 & 150--200 \\
\textit{P.\ macrocephalus} (sperm) & 40 & 0.65 & 2.50 & 307 & 1.24 & 1.6 & 0.55 & 60--70 \\
\textit{T.\ truncatus} (bottlenose) & 80 & 0.40 & 1.50 & 309 & 1.09 & 1.2 & 0.65 & 40--50 \\
\bottomrule
\end{tabular}
\end{table}

\subsection*{Synthesis across Clades}

Table~\ref{tab:phiC_factors} places all four clades side by side and
reveals the remarkable diversity of thermodynamic strategies that
converge on the same outcome: an extended entropy budget expressed
as $\Phi_C > 1$.
Primates and birds are both approximately single-state organisms from
the duty-cycle perspective---$\Phi_{\rm duty} \approx 1$ for primates
and $\Phi_{\rm duty} < 1$ for birds---yet they arrive at extended
lifespans by opposite routes.
Primates reduce the entropy cost of each individual cycle through
neural investment; birds overcome simultaneously adverse thermal and
kinematic conditions through mitochondrial and antioxidant
biochemistry.
Bats and cetaceans both achieve $\Phi_{\rm duty} > 1$, but by entirely
different mechanisms: bats exploit periodic whole-body hypothermia
that suppresses both the cardiac clock and biochemical damage rates
during hibernation, while cetaceans exploit a continuous vasovagal
bradycardia maintained throughout adult life with no thermal
suppression.
In each case, the prediction of $\Phi_C$ follows directly from
equation~\eqref{eq:PhiC_factored} using factor values drawn from
independently measured physiology, not from fits to lifespan data.

\begin{table}[H]
\centering\small
\renewcommand{\arraystretch}{1.40}
\caption{\textbf{Individual thermodynamic factors by clade.}
Representative species.
$+$: favourable ($>1$); $-$: adverse ($<1$); $=$: unity by definition.
$\Phi_{\rm duty}$ and $\Phi_{\rm thermal}$ are derived from measured
frequencies and temperatures and are independent of lifespan data.
$\Phi_{\rm neuro}$ requires the calibrated exponent $\alpha$
(equation~\ref{eq:alpha_def}) and is therefore not a parameter-free
prediction.
$\Phi_{\rm mito+oxid}$ and $\Phi_{\rm haz}$ are estimated from
independently published physiological and demographic data.}
\label{tab:phiC_factors}
\begin{tabular}{lcccccc}
\toprule
Clade (species) & $\Phi_{\rm duty}$ & $\Phi_{\rm thermal}$ &
$\Phi_{\rm neuro}$ & $\Phi_{\rm mito+oxid}$ & $\Phi_{\rm haz}$ &
$\Phi_C^{\rm pred}$ \\
\midrule
Primates (\textit{H.\,sapiens})
  & 1.00 ($=$) & 1.04 ($+$) & 2.51 ($+$) & --- & 1.00 & 2.60 \\
Bats (\textit{M.\,lucifugus})
  & 1.94 ($+$) & 4.10 ($+$) & --- & $\approx1$ & 0.68 & 5.39 \\
Birds (20\,g passerine)
  & 0.87 ($-$) & 0.73 ($-$) & --- & 2.33 ($+$) & 2.00 & 2.97 \\
Cetaceans (\textit{B.\,mysticetus})
  & 3.08 ($+$) & 1.17 ($+$) & --- & $\approx1$ & 0.35 & 1.76 \\
\bottomrule
\end{tabular}
\end{table}

\begin{table}[H]
\centering\small
\renewcommand{\arraystretch}{1.40}
\caption{\textbf{Predicted vs.\ observed clade multipliers $\Phi_C$.}
$\Phi_C^{\rm pred}$ from Table~\ref{tab:phiC_factors};
$\Phi_C^{\rm obs} = 10^{\Delta\ell}$ from the 230-species
dataset~\cite{taye_p1}.
$N_\star^{(C)}/N_0$ is the damage-equivalent budget after
duty-cycle correction.
The cetacean raw mean $\bar\ell = 8.80$ reflects an uncorrected
frequency; after full duty-cycle correction, the bowhead budget is
$N_\star/N_0 = 2.37$ (Section~\ref{sec:clade}).}
\label{tab:phiC}
\begin{tabular}{lcccc}
\toprule
Clade & $\Phi_C^{\rm pred}$ & $\Phi_C^{\rm obs}$ &
$N_\star^{(C)}/N_0$ & Primary driver \\
\midrule
Primates  & 2.60 & 2.40 & 2.6 & Neural entropy reduction \\
Bats      & 5.39 & 3.51 & 5.4 & Torpor $\times$ hypothermia \\
Birds     & 2.97 & 3.41 & 3.0 & Mitochondrial efficiency \\
Cetaceans & 1.76 & 0.64$^a$ & 2.4 & Bradycardic duty-cycle \\
\bottomrule
\end{tabular}
\end{table}

The cetacean entry in Table~\ref{tab:phiC} requires careful interpretation.
The raw observed multiplier $\Phi_C^{\rm obs} = 0.64$ lies below the
mammalian baseline, while the PBTE prediction after duty-cycle
correction is $\Phi_C^{\rm pred} = 1.76$---a factor-of-2.75
discrepancy, the largest of any clade.
The duty-cycle correction accounts for the thermodynamic
under-weighting of bradycardic beats and raises the
damage-equivalent budget to $N_\star^{(C)}/N_0 = 2.4$, but
substantial uncertainty remains in dive-fraction estimates across
species, in unquantified cetacean-specific antioxidant contributions,
and in the sensitivity of $\Phi_{\rm haz}$ to pre-industrial versus
modern mortality conditions.
The cetacean clade is the least well-predicted entry in the framework;
this is an open problem requiring species-resolved dive-fraction
telemetry and independent $\sigstar$ measurements.

Figures~\ref{fig:clade_phi_a}--\ref{fig:clade_scatter} display the
clade predictions graphically.

\begin{figure}[H]
\centering
\includegraphics[width=0.80\linewidth]{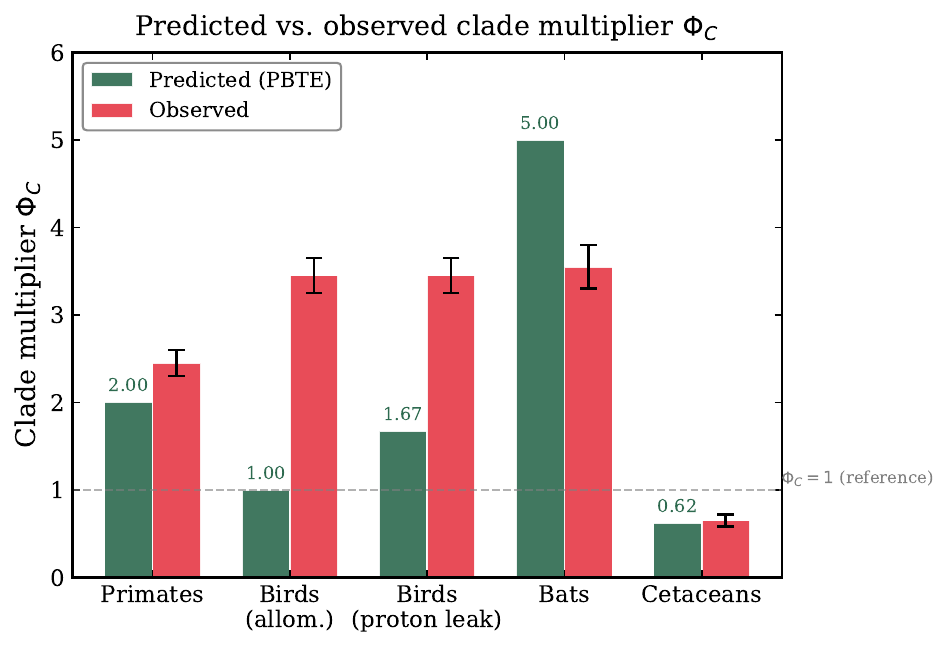}
\caption{\textbf{Predicted vs.\ observed clade multipliers $\Phi_C$.}
Grouped bar chart comparing PBTE-predicted (green) and empirically
observed (red, $\pm95\%$ CI) values of $\Phi_C$ for primates, birds,
bats, and cetaceans.
Both $\Phi_{\rm duty}$ and $\Phi_{\rm thermal}$ are adverse for birds;
their biochemical efficiency factor overcomes both.}
\label{fig:clade_phi_a}
\end{figure}

\begin{figure}[H]
\centering
\includegraphics[width=0.72\linewidth]{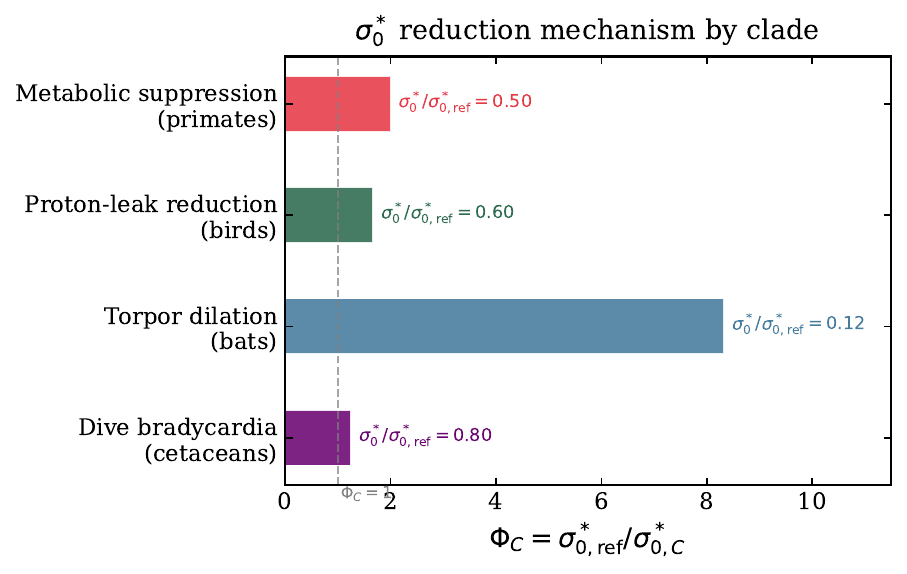}
\caption{\textbf{Physiological mechanism of $\sigstar$ reduction
by clade.}
Horizontal bar chart of $\Phi_C = {\sigstar}_{\rm ref}/{\sigstar}_C$
annotated with the fractional entropy-cost reduction.
Torpor time-dilation in bats produces the largest $\Phi_C$.}
\label{fig:clade_phi_b}
\end{figure}

\begin{figure}[H]
\centering
\includegraphics[width=0.80\linewidth]{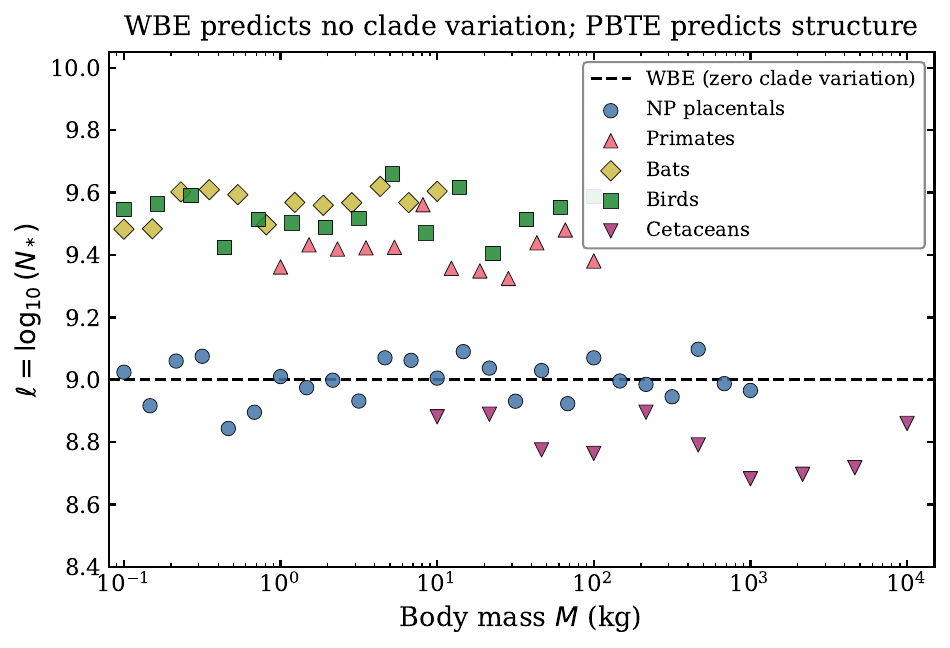}
\caption{\textbf{$\Nstar$ vs.\ body mass for five endotherm clades.}
$\ell = \log_{10}(\Nstar)$ plotted against body mass for 150 species.
Dashed line: WBE null prediction of zero inter-clade variation.
PBTE predicts flat within-clade profiles displaced vertically between
clades by $\Delta\ell = \log_{10}(\Phi_C)$.
The inter-clade offsets are the characteristic PBTE signature that
WBE cannot produce.}
\label{fig:clade_scatter}
\end{figure}

\section{Aging as Accumulated Internal Entropy}
\label{sec:aging}

The PBTE framework provides a natural interpretation of biological
aging as the progressive accumulation of internal entropy that
degrades the molecular machinery sustaining homeostasis, thereby
causing the biological clock to decelerate.
The model developed here is phenomenological: the damage sensitivity
parameter $\alpha$ and the damage accumulation rate $\gamma$
introduced below are not derivable from the PBTE framework itself
and must be treated as empirical inputs for each species.
The predictions of this section are therefore more speculative than
those of Sections~\ref{sec:derivation}--\ref{sec:clade} and should
be understood as a framework extension whose empirical calibration
remains to be carried out.

Let $S_{\rm int}(t)$ denote a coarse-grained measure of cumulative
macromolecular damage, epigenetic drift, and mitochondrial
dysfunction at chronological age $t$.
The simplest phenomenological model consistent with the
biological-proper-time framework is
\begin{equation}
  f(t) = f_0\exp\bigl(-\alpha\,S_{\rm int}(t)\bigr),
  \label{eq:aging_f}
\end{equation}
where $f_0$ is the youthful baseline cardiac frequency and $\alpha > 0$
is a species-specific damage sensitivity parameter.
If each independent damage event reduces metabolic capacity by a
factor $(1-\epsilon)$, then $n$ events give
$f \approx f_0(1-\epsilon)^n \approx f_0 e^{-\epsilon n}$, which is
precisely equation~\eqref{eq:aging_f} with $S_{\rm int} \propto n$.

Assuming linear entropy accumulation at rate $\gamma > 0$,
so that $S_{\rm int}(t) = \gamma t$, the biological clock decelerates
as $f(t) = f_0 e^{-\alpha\gamma t}$.
The remaining biological-time budget at age $t$ is
$\theta_{\rm rem}(t) = (f_0/\alpha\gamma)(e^{-\alpha\gamma t}
- e^{-\alpha\gamma L})$, and natural death occurs when this
remainder is exhausted, defining the natural lifespan
\begin{equation}
  L = -\frac{1}{\alpha\gamma}\ln\!\left(\frac{\alpha\gamma\,\Nstar}{f_0}\right).
  \label{eq:L_aging}
\end{equation}
The Gompertz--Makeham mortality hazard~\cite{gompertz1825} is
recovered by writing the failure probability as inversely proportional
to the remaining biological-time budget:
\begin{equation}
  h_{\rm fail}(t) \propto \frac{1}{\theta_{\rm rem}(t)}
  \;\approx\; e^{+\alpha\gamma t}
  \quad (t \ll L),
\end{equation}
which is the standard increasing-hazard Gompertz form with Makeham
exponent $\beta = \alpha\gamma > 0$.
The Gompertz exponent is therefore a thermodynamic prediction once
$\alpha$ and $\gamma$ are independently measured; it is not a free
fit parameter.

A subtle but important consistency point connects this aging model
to the PBTE invariant.
Equation~\eqref{eq:aging_f} with $\alpha > 0$ implies that the
organism arrives at natural death completing fewer beats per unit
chronological time than in youth.
This is consistent with the invariant $\theta_i(L_i) = \Nstar$,
which is a claim about the integral $\int_0^{L_i} f_i(t)\,\mathrm{d}t$
rather than about the instantaneous rate.
When $f_i(t)$ decelerates with age, the lifespan $L_i$ extends
correspondingly, so the integral still converges to $\Nstar$.
Clock deceleration trades biological-time rate for chronological-time
extent, leaving the total budget conserved.

Epigenetic aging clocks---the Horvath methylation
clock~\cite{horvath2013}, GrimAge, and related constructs---provide
species-specific empirical measures of $S_{\rm int}(t)$ at the
molecular level.
Within the PBTE framework, since $\theta$ advances at rate
$f_i(t) = f_0 e^{-\alpha S_{\rm int}(t)}$, the epigenetic clock rate
is proportional to $f_i(t)$.
Species with lower $\sigstar$, such as primates, should therefore
exhibit epigenetic clocks that run slower relative to chronological
age: their biological proper time advances more slowly per calendar
year, and molecular aging proceeds correspondingly more slowly per
calendar year as well.
Quantitatively, the PBTE framework predicts that the slope of
epigenetic age versus chronological age scales as
$\sigstar/{\sigstar}_{\rm ref}$ across species---a falsifiable
prediction testable with existing cross-species methylation databases
combined with species-resolved $\sigstar$ measurements.

\begin{figure}[H]
\centering
\includegraphics[width=0.72\linewidth]{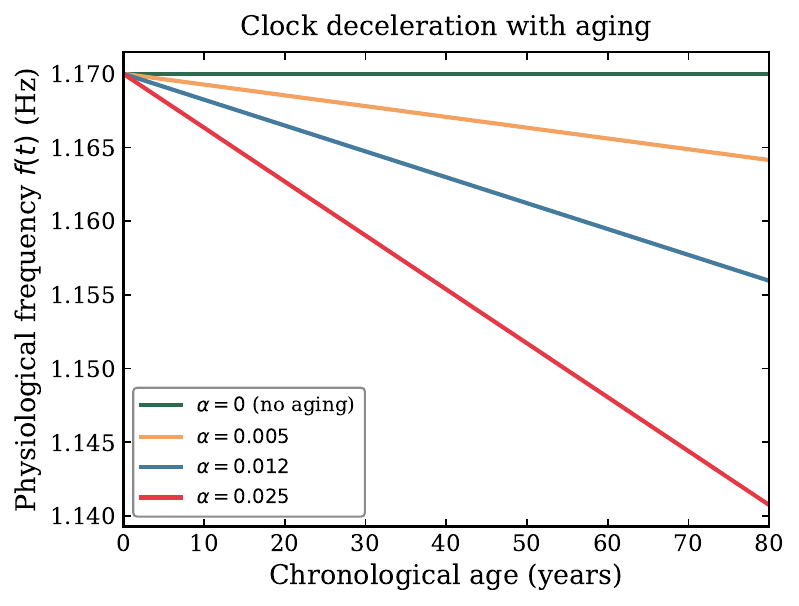}
\caption{\textbf{Physiological clock deceleration with aging.}
Resting frequency $f(t) = f_0\exp(-\alpha\gamma t)$ for
$f_0 = 1.17$\,Hz, $\gamma = 0.01$\,yr$^{-1}$, and
$\alpha \in \{0, 0.005, 0.012, 0.025\}$.
The value $\alpha \approx 0.012$ reproduces the empirically
observed decline of approximately 1\,bpm per decade in adult
resting heart rate.
The clock slows but the lifespan extends proportionally,
maintaining the PBTE invariant.}
\label{fig:clock_decel}
\end{figure}

\begin{figure}[H]
\centering
\includegraphics[width=0.72\linewidth]{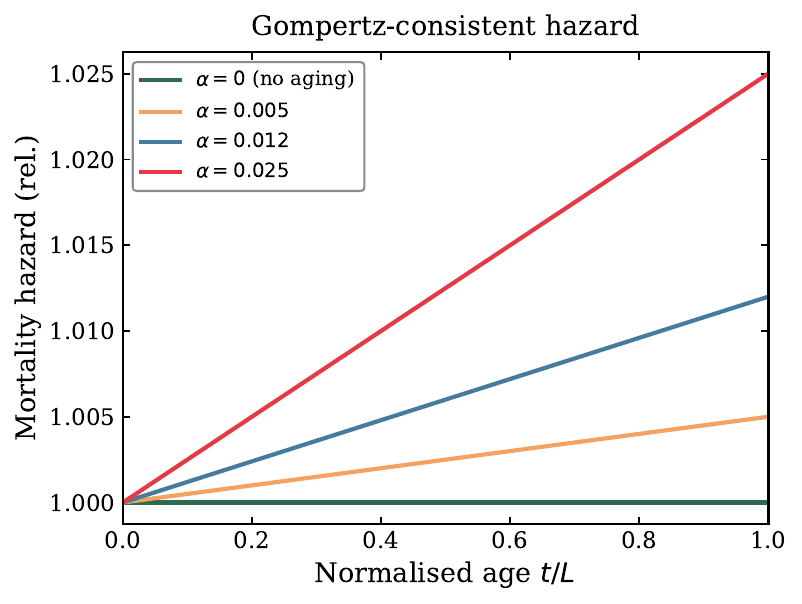}
\caption{\textbf{Gompertz-consistent mortality hazard from PBTE
clock deceleration.}
Relative mortality hazard versus normalised age $t/L$ for the same
four $\alpha$ values as Figure~\ref{fig:clock_decel}.
The hazard increases monotonically with age for $\alpha > 0$,
reproducing the Gompertz exponential form with rate $\beta = \alpha\gamma L$.
The Gompertz exponent is a thermodynamic consequence of the
damage-accumulation model, not a free fit parameter.}
\label{fig:gompertz}
\end{figure}

\begin{figure}[H]
\centering
\includegraphics[width=0.72\linewidth]{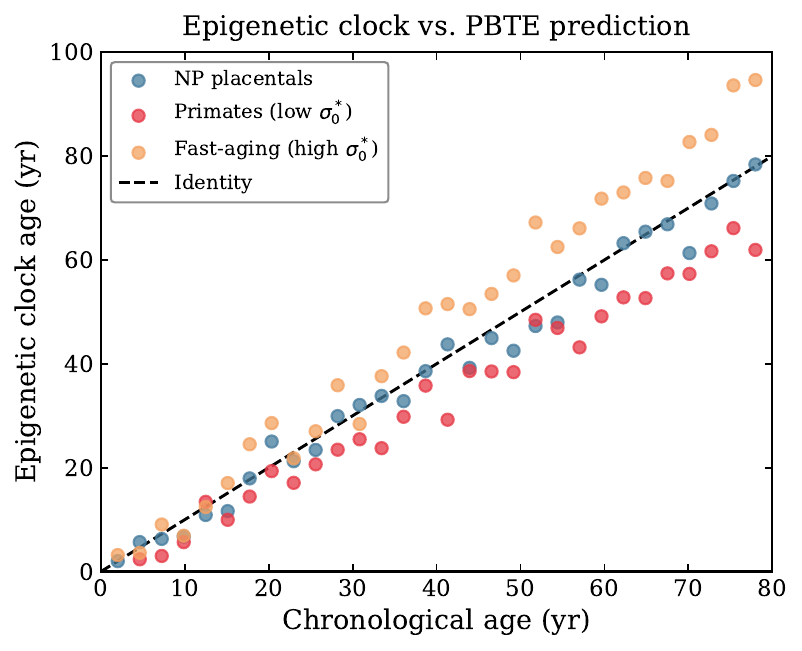}
\caption{\textbf{Epigenetic clock age vs.\ PBTE prediction (schematic).}
These data are simulated to illustrate the prediction; real
cross-species methylation data are required to test it.
Epigenetic clock age versus chronological age for three cohorts:
non-primate placentals (slope $\approx 1$); primates with
${\sigstar}_{\rm pri} = 0.50\,{\sigstar}_{\rm ref}$ (slope $\approx 0.80$,
biological age lags chronological); fast-aging cohort with elevated
$\sigstar$ (slope $\approx 1.18$).
PBTE predicts the epigenetic slope scales as
$\sigstar/{\sigstar}_{\rm ref}$.}
\label{fig:epigenetic}
\end{figure}


\section{Kinematic Layer: Defining Biological Proper Time}
\label{sec:kinematics}

\subsection{Biological proper time and the lifetime invariant}

We begin by introducing the central kinematic variable of the theory.

\begin{definition}[Biological proper time]
\label{def:bpt}
Let $t$ denote chronological time and let $f_i(t) > 0$ be the instantaneous
resting cardiac frequency of organism or species $i$. The biological proper
time accumulated by organism $i$ up to chronological time $t$ is defined by
\begin{equation}
  \thetabio_i(t) \;\equiv\; \int_0^t f_i(t')\,\mathrm{d}t'.
  \label{eq:theta}
\end{equation}
\end{definition}

This definition is fundamental because it separates two notions of time that
are ordinarily conflated. Chronological time measures the external passage of
duration, whereas biological proper time measures internally accumulated
physiological activity. Since $f_i(t)$ has dimensions of inverse time, the
quantity $\thetabio_i(t)$ is dimensionless. It is therefore not a duration in
the ordinary sense, but rather a count-like cumulative measure of biological
progress. In differential form, the defining relation becomes
\begin{equation}
  \mathrm{d}\thetabio_i = f_i(t)\,\mathrm{d}t,
\end{equation}
which plays the role of the basic kinematic law of the framework.

The interpretation is immediate. An organism is born at $\thetabio_i=0$ and,
within the PBTE picture, its life history unfolds as a trajectory in
biological-time space until a terminal value is reached. In the simplest form
of the theory this terminal value is the lifetime cycle budget $\Nstar$, so
that death occurs when
\begin{equation}
  \thetabio_i(L_i)=\Nstar.
\end{equation}
More generally, in the extended clade-sensitive formulation one writes
\begin{equation}
  \thetabio_i(L_i)=\Nstar=N_0\Phi_C,
\end{equation}
where $N_0$ is the canonical baseline invariant and $\Phi_C\ge 1$ is the clade
multiplier introduced in earlier papers.

When the resting frequency is approximated by a lifetime average
$\bar f_i$, equation~\eqref{eq:theta} reduces to the simple product
\begin{equation}
  \thetabio_i(L_i)=\bar f_i L_i \times 525{,}960,
  \label{eq:invariant}
\end{equation}
where the factor $525{,}960$ converts years into minutes. This is the
dimensionless lifetime count that underlies the empirical PBTE relation. For
non-hibernating, non-primate placental mammals, previous work showed that this
quantity is approximately constant and close to $10^9$. The observed scatter is
small on a logarithmic scale, corresponding to a factor of order unity across
species that differ enormously in body mass and lifespan. This is precisely the
empirical content of the lifetime invariant: small mammals live briefly but
accumulate biological proper time rapidly, whereas large mammals live much
longer but accumulate it slowly. The total lifetime budget remains nearly the
same.

This point is conceptually crucial. The invariance does not mean that all
species experience chronological time equally. Quite the opposite: their rates
of movement through biological time are drastically different. A mouse has a
much larger $f_H$ than an elephant and therefore advances through
$\thetabio$-time much more quickly per unit chronological time. The elephant,
with a much lower heart rate, accumulates biological proper time more slowly,
which is exactly why it requires a much longer chronological lifespan to reach
the same total biological budget. Thus the relevant comparison is not that the
elephant has ``more'' biological time, but that its biological clock runs more
slowly. This is the correct interpretation of the inverse frequency relation.

It is then natural to define a normalized biological age,
\begin{equation}
  \thetahat_i(t)=\frac{\thetabio_i(t)}{{\Nstar}_i},
  \qquad 0\le \thetahat_i \le 1,
  \label{eq:theta_hat_p5}
\end{equation}
which gives a universal fractional measure of progress through the biological
lifetime. In this representation, birth corresponds to $\thetahat_i=0$ and
death to $\thetahat_i=1$. The important prediction is that developmental and
aging landmarks should occur at approximately similar values of $\thetahat$
across species even when they occur at radically different chronological ages.
Sexual maturity, onset of senescence, and peak fecundity may therefore be
interpreted not merely as events occurring after a given number of months or
years, but as events located at characteristic fractions of the lifetime
biological trajectory. This is one of the deepest conceptual advantages of the
biological proper-time formulation: it expresses life history on an intrinsic
scale rather than an external one.

To quantify departures from the canonical mammalian baseline, we define
\begin{equation}
  \delta_i=\log_{10}\!\left(\frac{\thetabio_i(L_i)}{N_0}\right).
  \label{eq:delta}
\end{equation}
This logarithmic deviation provides a compact diagnostic of biological
exceptionality. Species with $|\delta_i|<0.15$ lie close to the baseline
invariant and may be regarded as ordinary realizations of the canonical
mammalian rule. Positive values of $\delta_i$ indicate an extended lifetime
budget relative to the baseline, pointing toward longevity-enhancing mechanisms
such as neuro-metabolic protection, torpor, hypothermia, mitochondrial
efficiency, or oxidative defense. Negative values indicate either elevated
extrinsic hazard, physiological suppression of the effective cycle budget, or,
in some cases, incomplete correction for duty cycle and metabolic state. In
this sense, $\delta_i$ functions as a phenomenological order parameter for
deviation from the standard PBTE manifold.

Table~\ref{tab:deviations} summarizes representative species and illustrates
the explanatory power of the framework.

\begin{table}[htbp]
\centering\small
\renewcommand{\arraystretch}{1.3}
\caption{\textbf{Species-level biological proper time and deviations from the
canonical invariant.}
$\thetabio(L)=\bar{f}_H \times 525{,}960 \times L$.
The deviation is defined by $\delta_i=\log_{10}(\thetabio(L)/N_0)$.
Values satisfying $|\delta_i|<0.15$ are consistent with the canonical
mammalian baseline. Positive deviations reflect identifiable longevity
mechanisms, whereas apparent negative deviations may arise from incomplete
physiological correction. Data sources include AnAge and Calder.}
\label{tab:deviations}
\begin{tabularx}{\textwidth}{>{\raggedright}p{5.0cm} r r r r X}
\toprule
\textbf{Species} & $\bar{f}_H$ (bpm) & $L$ (yr) &
$\thetabio(L)$ ($\times 10^9$) & $\delta_i$ & \textbf{Mechanism} \\
\midrule
\textit{Mus musculus} (mouse)          & 600 & 3.2  & 1.01 & $+0.01$ & Baseline \\
\textit{Rattus norvegicus} (rat)       & 420 & 4.5  & 0.99 & $-0.01$ & Baseline \\
\textit{Oryctolagus cuniculus} (rabbit)& 205 & 9.3  & 1.00 & $+0.00$ & Baseline \\
\textit{Canis lupus familiaris} (dog)  & 100 & 20   & 1.05 & $+0.02$ & Baseline \\
\textit{Loxodonta africana} (elephant) &  30 & 65   & 1.03 & $+0.01$ & Baseline \\
\textit{Homo sapiens} (human)          &  70 & 79   & 2.91 & $+0.46$ & Neuro-metabolic ($\Phi_{\rm neuro}$) \\
\textit{Myotis lucifugus} (bat)        & 155 & 34   & 2.77 & $+0.44$ & Torpor + hypothermia \\
\textit{Diomedea exulans} (albatross)  & 110 & 65   & 3.75 & $+0.57$ & Mitochondrial + antioxidant \\
\textit{Balaena mysticetus} (bowhead)  &   8 & 180  & 0.76$^\dagger$ &
$-0.12^\dagger$ & Diving (raw; $\Phi_{\rm duty}$ corrects) \\
\bottomrule
\multicolumn{6}{l}{\footnotesize $^\dagger$Uncorrected raw count.
With $\Phi_{\rm duty}=3.08$: ${\Nstar}^{(\rm bow)}=0.76 \times 3.08 \times 10^9
=2.34\times10^9$ ($\delta=+0.37$).}
\end{tabularx}
\end{table}

\subsection{Choice of rate function, normalization, and empirical estimation}

The definition in equation~\eqref{eq:theta} is mathematically general: any
strictly positive physiological rate could in principle be used to define a
biological proper time. The issue is therefore not mathematical possibility but
biological suitability. The preferred rate function should capture global
metabolic pace, should be empirically measurable with manageable uncertainty,
and should remain meaningful across different physiological states. Resting
cardiac frequency satisfies these requirements especially well in endotherms.
It correlates strongly with metabolic throughput, can be measured directly or
compiled from comparative databases, and remains interpretable across
normothermic, torpid, and hibernating regimes. This makes it particularly
well-suited for the kinematic role assigned to it in PBTE.

For cross-species comparison, however, raw heart rate alone is not sufficient.
Differences in body temperature and body mass shift kinetic pace in systematic
ways and must be normalized if one wishes to compare species on a common
physiological scale. The corrected rate may therefore be written as
\begin{equation}
  f_{\rm corr}(t)=f(t)
  \exp\!\left[\frac{E_a}{k_B}\!\left(\frac{1}{T_{\rm ref}}-\frac{1}{T(t)}\right)\right]
  \left(\frac{M_{\rm ref}}{M(t)}\right)^{1/4},
  \label{eq:f_corr}
\end{equation}
where $E_a \approx 0.65\,\mathrm{eV}$ is an effective activation energy,
$T_{\rm ref}=310\,\mathrm{K}$ is the reference temperature, and
$M_{\rm ref}=1\,\mathrm{kg}$ is the reference mass. The thermal factor is
Arrhenius-like and reflects the same molecular-kinetic origin that appears in
the thermal multiplier introduced in the extended PBTE theory. The mass factor
encodes the standard quarter-power scaling of physiological time. Thus
equation~\eqref{eq:f_corr} is not an ad hoc correction but an expression of the
same kinetic logic that underlies the larger framework.

In practice, empirical estimation of $\thetabio_i$ proceeds in a direct way.
One first selects an appropriate physiological rate, usually resting cardiac
frequency, and determines whether cross-species normalization is required.
Using comparative datasets such as AnAge or PanTHERIA, one then estimates the
species mean rate $\bar f_i$ and the corresponding lifespan $L_i$. Their
product, multiplied by the appropriate unit-conversion factor, yields the total
biological proper time accumulated over life. Comparison with the canonical
baseline $N_0$ or the clade-corrected value $N_0\Phi_C$ then reveals whether
the species belongs to the standard invariant class or requires a longevity
multiplier for explanation. Finally, the quantity $\delta_i$ provides a compact
numerical summary of that deviation.

The significance of this construction extends beyond bookkeeping. By replacing
chronological duration with an internally accumulated physiological coordinate,
the theory provides a genuine biological geometry of life history. Different
species trace trajectories through chronological time at different speeds, but
their progression through biological proper time exhibits a striking regularity.
This is why PBTE is not merely an empirical scaling rule but a candidate
kinematic principle of comparative biology. It states that what is conserved
across a wide range of organisms is not lifespan itself, and not heart rate
itself, but the integrated physiological progression encoded in
$\thetabio_i(L_i)$.

\section{Thermodynamic Layer: Entropy Accumulates Uniformly in $\thetabio$}
\label{sec:thermo}

\subsection{Thermodynamic closure in biological proper time}

Living systems persist in a far-from-equilibrium state by continuously
exporting entropy to their surroundings. This idea, originally expressed
by \textit{Erwin Schr\"odinger} in his formulation of organisms as systems
that ``feed on negative entropy''~\cite{schrodinger1944}, was later given a
precise quantitative foundation within the framework of non-equilibrium
thermodynamics developed by \textit{Ilya Prigogine}~\cite{prigogine1967}.
In this perspective, the defining feature of life is not equilibrium
structure, but sustained entropy production coupled to energy throughput.

Within the PBTE framework, the entropy production rate of organism $i$
is expressed as a closure relation,
\begin{equation}
  \dot{e}_{p,i} = \sigma_{0,i}\, f_i,
\end{equation}
which states that entropy production is proportional to the intrinsic
physiological rate $f_i$, with $\sigma_{0,i}$ representing the entropy
generated per biological cycle. This relation encodes a central physical
assumption: each fundamental physiological cycle carries a characteristic
thermodynamic cost that is approximately constant for a given organism.

The significance of biological proper time becomes transparent when one
rewrites the total entropy production over the lifespan in terms of the
intrinsic coordinate $\thetabio_i$. Using the kinematic relation
$\mathrm{d}\thetabio_i = f_i\,\mathrm{d}t$, one obtains
\begin{equation}
  \Sigma_i = \int_0^{L_i} \dot{e}_{p,i}\,\mathrm{d}t
  = \int_0^{{\Nstar}_i} \frac{\sigma_{0,i} f_i}{f_i}\,\mathrm{d}\thetabio_i
  = \sigma_{0,i}\,{\Nstar}_i.
  \label{eq:entropy_theta}
\end{equation}
This transformation is not merely a change of variables; it reveals a
structural simplification. In chronological time, entropy production is
rate-dependent and varies across species and physiological states. In
biological proper time, however, the dependence on $f_i$ cancels exactly,
and the entropy budget becomes separable and linear.

Differentiating with respect to $\thetabio_i$ gives
\begin{equation}
  \frac{\mathrm{d}\Sigma_i}{\mathrm{d}\thetabio_i}
  = \sigma_{0,i} = \mathrm{const}.
  \label{eq:uniform_entropy}
\end{equation}
This result expresses the thermodynamic content of the theory in its
clearest form. Entropy accumulates at a constant rate when measured in
biological proper time, independent of how fast or slow the organism moves
through chronological time. The coordinate $\thetabio_i$ is therefore not
an arbitrary construction; it is the unique intrinsic parameter for which
the thermodynamic evolution becomes uniform. In this sense, it plays a role
directly analogous to proper time in relativistic dynamics, where physical
laws simplify when expressed in an intrinsic frame.

The lifetime relation $\Sigma_i = \sigma_{0,i}{\Nstar}_i$ immediately yields
\begin{equation}
  {\Nstar}_i = \frac{\Sigma_i}{\sigma_{0,i}},
\end{equation}
which provides a thermodynamic interpretation of the PBTE invariant.
The lifetime cycle budget is not an empirical coincidence, but the ratio
between the total entropy exported over life and the entropy cost per
physiological cycle. This identity is the bridge connecting metabolism,
entropy production, and lifespan.

\subsection{Thermodynamic classification of longevity mechanisms}

The relation $\Nstar = \Sigma/\sigma_0$ reveals that chronological lifespan
can be extended through two fundamentally distinct thermodynamic routes.
These routes are not merely descriptive categories; they correspond to
different physical manipulations of the underlying entropy balance.

The first route operates by altering the rate at which biological proper
time accumulates, while leaving the total lifetime budget unchanged. In
this case, the organism reduces its physiological pace $f_i$, so that
\begin{equation}
  \frac{\mathrm{d}\thetabio_i}{\mathrm{d}t} = f_i \;\;\text{decreases},
  \qquad \Nstar = \text{constant}.
\end{equation}
Because the total number of cycles remains fixed, a lower rate implies a
longer chronological duration required to reach the same terminal value.
The lifespan therefore increases according to
\begin{equation}
  L = \frac{\Nstar}{f_i \times 525{,}960}.
  \label{eq:class1}
\end{equation}
This mechanism may be interpreted as a form of biological time dilation.
The organism does not gain additional biological time; instead, it
traverses biological time more slowly. Hibernation in bats, diving
bradycardia in marine mammals, and caloric restriction in laboratory
models are all examples of this class. In each case, the system reduces
its instantaneous metabolic pace while preserving the total cycle budget.

The second route operates in a fundamentally different way. Instead of
slowing the passage through biological time, the organism alters the
thermodynamic cost of each cycle. A reduction in $\sigma_{0,i}$ lowers the
entropy generated per cycle, so that the same total entropy budget
corresponds to a larger number of cycles,
\begin{equation}
  \sigma_{0,i} \;\;\text{decreases},
  \qquad \Nstar = \frac{\Sigma}{\sigma_{0,i}} \;\;\text{increases}.
  \label{eq:class2}
\end{equation}
In this case, lifespan extension arises from an expansion of the biological
budget itself rather than a slowing of its accumulation. This mechanism is
associated with improved energetic efficiency and reduced dissipation at
the cellular or organismal level. Neural investment in primates,
mitochondrial coupling efficiency in birds, and enhanced oxidative
protection mechanisms are representative examples. Here, the organism does
not merely live more slowly; it lives more efficiently, accumulating less
entropy per cycle and therefore sustaining a greater total number of
cycles over its lifetime.

The distinction between these two mechanisms is central to the PBTE
framework. Class~1 modifies the mapping between chronological time and
biological time, whereas Class~2 modifies the thermodynamic cost structure
that defines the biological time budget itself.

\subsection{Experimental discrimination via epigenetic clocks}

A key strength of the thermodynamic formulation is that the two classes of
mechanism lead to experimentally distinguishable predictions. This can be
seen by considering epigenetic aging clocks, which provide an empirical
measure of biological age through DNA methylation patterns, as introduced
by \textit{Steve Horvath}~\cite{horvath2013}.

Define the epigenetic aging rate with respect to biological proper time,
\begin{equation}
  \dot{A}_{\rm epi}
  = \frac{\mathrm{d}(\text{methylation age})}{\mathrm{d}\thetabio_i}
  = \frac{1}{\sigma_{0,i}}
    \frac{\mathrm{d}(\text{methylation age})}{\mathrm{d}\Sigma_i}.
\end{equation}
This expression shows that epigenetic aging per unit biological time is
directly tied to entropy production per cycle.

For Class~1 mechanisms, where $\sigma_{0,i}$ remains unchanged and only the
rate $f_i$ is reduced, the epigenetic aging rate per unit biological proper
time is predicted to remain invariant. The organism accumulates biological
time more slowly in chronological units, but each unit of biological time
carries the same thermodynamic cost. Consequently, epigenetic age per
heartbeat should be identical between treated and control populations.

For Class~2 mechanisms, where $\sigma_{0,i}$ is reduced, the situation is
different. A lower entropy cost per cycle implies that epigenetic aging per
unit biological proper time must also decrease. In this case, the rate
$\dot{A}_{\rm epi}$ is reduced by a factor proportional to
$\Phi_C^{-1}$, reflecting the expanded biological budget.

These predictions provide a direct experimental test. Under caloric
restriction, which is a canonical Class~1 intervention, epigenetic aging
per heartbeat should remain unchanged relative to ad-libitum controls, as
observed in primate studies~\cite{colman2014}. In contrast, species or
conditions associated with enhanced thermodynamic efficiency should exhibit
reduced epigenetic aging per cycle. Thus the PBTE framework does not merely
classify longevity mechanisms; it generates falsifiable predictions that
can be tested using modern molecular biomarkers.

\section{Geometric Layer: The Biological Metric}
\label{sec:geometry}

\subsection{Biological metric, arc length, and the geometric meaning of PBTE}

Geometry provides a language in which conservation laws often appear in
their most transparent form. Whenever a physical quantity can be expressed
as the arc length of a trajectory, it acquires a natural invariance
structure, together with notions of optimality, deformation, and intrinsic
measurement. The PBTE invariant admits precisely such a geometric
interpretation.

We introduce the biological line element
\begin{equation}
  \mathrm{d}s_i^2 = \varphi\bigl(f_i(t)\bigr)\,\mathrm{d}t^2,
  \label{eq:line_element}
\end{equation}
where $\varphi(f)>0$ is a metric function defined on the space of
physiological rates. This definition is general and allows different
choices depending on what aspect of biological dynamics one wishes to
emphasize. The canonical and physically distinguished choice is
\begin{equation}
  \varphi(f_i)=f_i^2,
\end{equation}
for which the line element reduces to
\begin{equation}
  \mathrm{d}s_i = f_i(t)\,\mathrm{d}t = \mathrm{d}\thetabio_i.
  \label{eq:canonical_metric}
\end{equation}

This identification is fundamental. Biological arc length is identical to
biological proper time. The accumulated length of the life trajectory in
this metric is therefore
\begin{equation}
  s_i(L_i)=\int_0^{L_i} f_i(t)\,\mathrm{d}t = {\Nstar}_i,
  \label{eq:arc_length}
\end{equation}
which recovers the PBTE invariant as a purely geometric statement. The
lifetime of an organism is not characterized by its chronological duration
but by the intrinsic arc length of its trajectory in biological time.
Different species traverse this trajectory at different speeds, yet the
total length remains approximately fixed. A mouse and a bowhead whale
therefore follow curves of equal arc length in the biological metric, even
though their parametrization by chronological time differs dramatically.
This is the precise geometric content of the invariant.

\subsection{Intrinsic time and the controlled spacetime analogy}

The role played by $\thetabio_i$ invites comparison with proper time in
relativistic physics, but the analogy must be stated carefully to avoid
misinterpretation. In special relativity, the Minkowski interval defines an
intrinsic measure of time along a worldline that is invariant under Lorentz
transformations. In the present context, the biological metric defines an
intrinsic temporal coordinate that is invariant under reparametrizations of
chronological time. The quantity $\thetabio_i$ therefore functions as a
scalar measure of physiological progression that is independent of the
units in which external time is measured.

The similarity is structural rather than dynamical. There is no Lorentz
symmetry underlying the biological metric, and no universal limiting speed
analogous to the speed of light. The invariant $\Nstar$ is not a fundamental
constant of nature but a statistical quantity with clade-dependent
modulation. Furthermore, the biological metric is purely temporal and does
not involve any mixing of spatial and temporal coordinates. The analogy is
therefore useful only insofar as it emphasizes the existence of an intrinsic
time variable in which the dynamics simplifies; it should not be extended
beyond this limited scope.

\subsection{Metabolic frame transformations and invariance}

The geometric formulation naturally leads to a notion of transformation
between organisms viewed as different ``metabolic frames.'' Requiring that
a given physiological event correspond to the same increment of biological
arc length in two organisms implies
\begin{equation}
  \mathrm{d}t_j = \frac{\bar{f}_i}{\bar{f}_j}\,\mathrm{d}t_i.
  \label{eq:frame_transform}
\end{equation}
This relation expresses how chronological time must be rescaled so that the
same intrinsic biological increment is obtained. Integrating this relation
over the full lifespan yields
\begin{equation}
  \bar{f}_i L_i \approx \bar{f}_j L_j \approx \Nstar,
\end{equation}
which is simply the PBTE invariant written as a geometric identity. In this
sense, the invariance of $\Nstar$ is equivalent to invariance of arc length
under the class of metabolic transformations that rescale chronological
time while preserving biological progression.

\subsection{Nonstationarity as geometric curvature of the trajectory}

Real biological systems are not stationary. The physiological rate $f_i(t)$
changes over development, aging, and environmental perturbation. These
changes can be captured geometrically by examining the variation of the
metric along the trajectory. Defining
\begin{equation}
  R_i(t) = -\frac{1}{2}\frac{\mathrm{d}^2}{\mathrm{d}t^2}
  \ln\bigl[\varphi\bigl(f_i(t)\bigr)\bigr],
  \label{eq:R_nonstationarity}
\end{equation}
one obtains a scalar measure of nonstationarity with dimensions of
inverse time squared. When $f_i(t)$ is constant, the metric is uniform and
$R_i=0$. Deviations from zero quantify curvature-like effects in the
temporal trajectory.

Large positive values of $R_i$ correspond to accelerating biological time,
where the organism moves increasingly rapidly through its physiological
trajectory. Large negative values correspond to deceleration. These regions
are expected to correlate with periods of elevated thermodynamic activity,
such as rapid growth, stress responses, or pathological transitions.
Thus $R_i(t)$ provides a geometric diagnostic of physiological instability,
linking time-dependent dynamics to measurable biological states.

\subsection{Action functional and exact recovery of entropy production}

The geometric structure becomes fully consistent with the thermodynamic
layer when one defines a lifetime action functional by weighting the arc
length with the entropy cost per cycle,
\begin{equation}
  \mathcal{A}_i = \sigma_0 \int_0^{L_i} \sqrt{\varphi(f_i(t))}\,\mathrm{d}t.
  \label{eq:action}
\end{equation}
For the canonical choice $\varphi=f_i^2$, this reduces exactly to
\begin{equation}
  \mathcal{A}_i = \sigma_0 \int_0^{L_i} f_i(t)\,\mathrm{d}t
  = \sigma_0 {\Nstar}_i = \Sigma_i,
\end{equation}
which coincides with the total entropy produced over the lifespan. This is
a strong consistency result: the geometric arc length, when properly
weighted, reproduces the thermodynamic invariant derived earlier. Geometry
and thermodynamics are therefore not separate descriptions but two
representations of the same underlying structure.

\subsection{Information geometry and minimal dissipation trajectories}

The geometric picture can be extended further by considering the space of
internal physiological states. Let $p_i(x;\xi)$ denote the probability
distribution over microscopic configurations, parametrized by macroscopic
variables $\xi$. The Fisher information metric,
\begin{equation}
  g_{\mu\nu}(\xi) = \int p_i(x;\xi)
  \frac{\partial\ln p_i}{\partial\xi^\mu}
  \frac{\partial\ln p_i}{\partial\xi^\nu}\,\mathrm{d}x,
  \label{eq:fisher}
\end{equation}
defines a second, independent notion of distance, measuring how much the
internal state of the organism changes rather than how many cycles are
executed.

For systems evolving near a non-equilibrium steady state, the excess
entropy production rate satisfies the bound
\begin{equation}
  \dot{S}_i^{\rm exc}(t) \ge \frac{k_B}{4\tau_{\rm rel}}
  \, g_{\mu\nu}(\xi)\,\dot\xi^\mu \dot\xi^\nu,
  \label{eq:fisher_entropy}
\end{equation}
linking thermodynamic dissipation to motion in information space. This
relation suggests a natural optimality principle. Trajectories that
minimize the Fisher arc length,
\begin{equation}
  \delta \int_0^{L_i}
  \sqrt{g_{\mu\nu}\,\dot\xi^\mu\dot\xi^\nu}\,\mathrm{d}t \approx 0,
  \label{eq:geodesic}
\end{equation}
also minimize excess entropy production.

This leads to the geodesic optimality hypothesis: physiological regulation
tends to drive the system along approximately geodesic paths in information
space, thereby reducing dissipative losses. Deviations from these optimal
paths should manifest as increased entropy production and should correlate
with measurable biological indicators such as oxidative stress, chronic
inflammation, and accelerated aging. This prediction provides a direct
bridge between abstract geometric structure and experimentally observable
biomarkers.

\section{Relativistic Layer: The Transformation Group and Metabolic Spectrum}
\label{sec:relativity}

\subsection{Biological transformation law and its physical meaning}

Consider two organisms $i$ and $j$ evolving within the same chronological
time coordinate $t$. Their biological proper times are defined through
$\mathrm{d}\thetabio_i = f_i(t)\,\mathrm{d}t$ and
$\mathrm{d}\thetabio_j = f_j(t)\,\mathrm{d}t$. Taking the ratio of these
increments gives
\begin{equation}
  \frac{\mathrm{d}\thetabio_i}{\mathrm{d}\thetabio_j}
  = \frac{f_i(t)}{f_j(t)}.
  \label{eq:transform_law}
\end{equation}

This relation is the biological transformation law. It expresses how two
organisms calibrate their intrinsic clocks relative to one another at the
same chronological instant. The ratio $f_i/f_j$ is therefore not merely a
numerical comparison of heart rates; it is the Jacobian relating two
intrinsic temporal parametrisations of the same external time axis.

The interpretation must be stated carefully. A larger $f_i$ does not mean
that organism $i$ possesses more biological time, but that it traverses
biological time more rapidly. Thus, if $f_i > f_j$, then
$\thetabio_i$ advances faster than $\thetabio_j$ with respect to the same
chronological increment. The distinction between rate and total accumulated
time is essential and removes the ambiguity that often appears in naive
interpretations of interspecies comparisons.

\subsection{Group structure and invariance}

The transformation law possesses a simple algebraic structure. For three
organisms $i$, $j$, and $k$, one finds
\begin{equation}
  \frac{\mathrm{d}\thetabio_i}{\mathrm{d}\thetabio_k}
  = \frac{f_i}{f_k}
  = \frac{f_i}{f_j}\cdot\frac{f_j}{f_k},
  \label{eq:group_composition}
\end{equation}
showing that transformations compose multiplicatively. The set of all
such ratios forms a one-dimensional abelian group under multiplication,
isomorphic to $(\mathbb{R}_{>0},\times)$. This group may be interpreted as
the group of metabolic time rescalings.

Within this structure, the PBTE invariant $\Nstar$ plays the role of a
conserved scalar. While chronological time differs between organisms and
their intrinsic parametrisations differ by elements of the scaling group,
the total biological arc length remains approximately invariant. In this
sense, $\Nstar$ is the quantity that is preserved under metabolic frame
transformations, just as proper time is invariant under changes of
parametrisation in relativistic kinematics.

\subsection{Quantitative illustration: mouse and elephant}

The transformation law acquires concrete meaning when applied to
well-characterised species. For a mouse with
$\bar{f}_H \approx 600\,\mathrm{bpm}$ and an elephant with
$\bar{f}_H \approx 30\,\mathrm{bpm}$, one obtains
\begin{equation}
  \frac{\mathrm{d}\thetabio_m}{\mathrm{d}\thetabio_E}
  = \frac{600}{30} = 20.
  \label{eq:mouse_elephant}
\end{equation}
Thus, per unit chronological time, the mouse advances through biological
proper time twenty times faster than the elephant.

Using the geometric relation established earlier,
\begin{equation}
  L = \frac{\Nstar}{f \times 525{,}960},
\end{equation}
one immediately obtains the inverse scaling of lifespans,
\begin{equation}
  \frac{L_E}{L_m} \approx \frac{f_m}{f_E} = 20,
\end{equation}
which yields $L_E \approx 64\,\mathrm{yr}$ for a mouse lifespan of
$3.2\,\mathrm{yr}$, in excellent agreement with observed elephant
lifespans of approximately $65\,\mathrm{yr}$.

The key point is conceptual. The mouse does not experience a shorter
biological lifetime in intrinsic units; rather, it traverses the same
biological arc length $\Nstar \sim 10^9$ at a much higher rate. The
difference between species lies in the speed along the trajectory, not in
the total distance covered. This is the precise meaning of the PBTE
invariant within the transformation framework.

\subsection{The metabolic spectrum and scaling structure}

To organise the wide range of biological rates, it is useful to introduce a
dimensionless scaling parameter
\begin{equation}
  \gamma_i = \frac{f_i}{f_{\rm ref}},
  \qquad f_{\rm ref} = 0.5\,\mathrm{Hz},
  \label{eq:gamma}
\end{equation}
where the reference value corresponds to the resting cardiac frequency of
a large mammal such as an elephant.

The biological realisation of $\gamma_i$ spans several orders of magnitude,
from extremely slow replication cycles in viruses to rapid metabolic
cycling in small endotherms. Viral systems correspond to
$\gamma \sim 10^{-5}$--$10^{-4}$, reflecting replication times of hours to
days. Bacterial systems occupy the range
$\gamma \sim 10^{-3}$--$10^{-2}$, consistent with division times on the
scale of tens of minutes to hours. Large mammals cluster near
$\gamma \sim 1$, while humans lie slightly above this scale with
$\gamma \sim 2$. Small mammals such as mice and shrews extend into the
range $\gamma \sim 10$--$10^2$, reflecting their high metabolic rates.

This broad distribution defines what may be called the metabolic-temporal
spectrum: a continuous scaling structure that organises biological systems
according to their intrinsic rates of physiological progression. Each
organism corresponds to a trajectory parametrised by its own value of
$\gamma_i$, and the ensemble of all such trajectories forms a structured
region in the space of biological evolution.

\subsection{The biological light-cone as an empirical envelope}

The collection of all biologically realised trajectories defines a
fan-shaped region in the space of chronological versus biological time.
This region may be referred to as the biological light-cone. The analogy is
again structural rather than literal. In relativistic physics, the
light-cone defines a strict causal boundary enforced by the invariance of
the speed of light. In the biological context, no such fundamental bound
exists.

Instead, the biological light-cone is an empirical envelope determined by
biochemical, physiological, and ecological constraints. It delineates the
range of metabolic rates that are viable for living systems. Organisms
within this envelope exhibit stable physiological operation, whereas
trajectories that would lie far outside it are not realised in nature.
Thus the biological light-cone is best understood as a domain of
feasibility rather than a causal structure.

The introduction of this envelope provides a unifying picture of the
metabolic spectrum. It situates individual organisms within a global
structure and highlights the role of scaling laws in constraining
biological diversity. At the same time, it emphasizes that the PBTE
invariant operates across this entire spectrum: despite vast differences
in metabolic rate, organisms traverse approximately the same intrinsic
biological distance over their lifetimes.

\section{Pathological Layer: The Temporal Order Parameter}
\label{sec:pathology}

\subsection{Temporal order parameter and thermodynamic meaning}

To quantify the progression of an organism through its intrinsic lifetime,
we introduce the temporal order parameter
\begin{equation}
  \zeta(t) = \frac{N(t)}{\Nstar}
  = \frac{1}{\Nstar}\int_0^t f(t')\,\mathrm{d}t'
  = \frac{\Sigma(t)}{\Sigma_\star},
  \label{eq:zeta}
\end{equation}
which simultaneously represents the fraction of the biological cycle
budget and the fraction of the lifetime entropy budget consumed up to
chronological time $t$. This dual representation is not coincidental; it
follows directly from the thermodynamic identity
$\Sigma = \sigma_0 \Nstar$, and therefore encodes both kinematic and
thermodynamic information in a single scalar quantity.

The interpretation is immediate. At birth, $\zeta(0)=0$, and for a system
evolving along the reference PBTE manifold one has $\zeta(L)=1$. The value
of $\zeta$ therefore provides a dimensionless measure of biological age
that is intrinsic and comparable across species. Unlike chronological age,
it reflects how much of the organism’s thermodynamic and physiological
budget has been expended.

Three qualitative regimes naturally emerge from this definition. When
$\zeta$ remains close to unity at the end of life, the organism follows a
homeostatic trajectory consistent with the canonical invariant. When
$\zeta$ exceeds unity, the organism has progressed too rapidly through its
biological budget, corresponding to accelerated aging and premature
depletion of physiological resources. When $\zeta$ remains below unity,
the organism has progressed more slowly, corresponding to extended
longevity and delayed aging. These regimes are not arbitrary categories but
direct consequences of the ratio structure defining $\zeta$.

\subsection{Empirical interpretation across physiological conditions}

The order parameter provides a unified interpretation of a wide range of
empirical observations. Elevated resting heart rate in mammals is well
known to correlate with reduced lifespan. Within the present framework,
this corresponds to an increase in the instantaneous rate $f(t)$, which
causes the integral $\int_0^t f(t')\,\mathrm{d}t'$ to grow more rapidly,
driving $\zeta(t)$ toward unity at an earlier chronological time. Lifespan
compression is therefore a direct consequence of accelerated traversal of
biological proper time.

Chronic inflammation produces a similar effect through a different
mechanism. Systemic inflammatory states elevate metabolic demand and
increase entropy production, effectively increasing the rate of biological
time accumulation. Epigenetic clocks, which track biological age through
methylation patterns, are observed to advance more rapidly under such
conditions. This can be interpreted as an increase in $\mathrm{d}\zeta /
\mathrm{d}t$, reflecting accelerated consumption of the entropy budget.

In contrast, hibernation and torpor provide a clear example of the opposite
regime. During torpor, heart rate and metabolic activity are suppressed by
orders of magnitude, drastically reducing $f(t)$. As a result, biological
proper time accumulates very slowly, and the chronological time required to
reach $\zeta=1$ is extended. This corresponds to $\zeta_{\rm life}<1$ at
comparable chronological ages, placing the organism in the extended
longevity regime.

These examples illustrate that $\zeta(t)$ is not merely a theoretical
construct but a quantity with direct physiological interpretation across a
range of biological states.

\subsection{Dynamics near the reference manifold}

To describe deviations from the canonical trajectory, define the fractional
deviation
\begin{equation}
  \phi(t) = \zeta(t) - \frac{t}{L_{\rm exp}},
\end{equation}
which measures the difference between the actual trajectory and the
reference linear progression expected under steady conditions.

Near the reference manifold, the dynamics of $\phi(t)$ may be approximated
by a linear stochastic differential equation of Ornstein--Uhlenbeck type,
\begin{equation}
  \dot\phi(t) = -\frac{\phi(t)}{\tau_\zeta} + \eta(t) + u(t),
  \label{eq:OU}
\end{equation}
where $\tau_\zeta$ is a characteristic resilience timescale,
$\eta(t)$ represents stochastic physiological fluctuations, and $u(t)$
represents sustained external or internal forcing, such as chronic stress
or therapeutic intervention.

The stationary power spectral density associated with this process is
\begin{equation}
  S_\phi(\omega) = \frac{2D_\eta \tau_\zeta^2}{1 + \omega^2 \tau_\zeta^2},
  \label{eq:PSD}
\end{equation}
which has a Lorentzian form with characteristic frequency
$\omega_c = 1/\tau_\zeta$. As the system approaches loss of stability,
$\tau_\zeta$ increases and the spectrum shifts toward lower frequencies,
producing $1/f$-like behaviour. This provides a direct and testable
prediction: physiological time series should exhibit critical slowing down
as the system approaches pathological transitions.

\subsection{Spatial extension and local temporal symmetry breaking}

In multicellular systems, the temporal order parameter need not be uniform
across space. Local tissues can evolve with intrinsic rates that differ
from the organism-wide average. Introducing a spatially resolved order
parameter,
\begin{equation}
  \zeta(x,t) = \frac{f(x,t)\,\tau(x,t)}{\Nstar},
  \label{eq:local_zeta}
\end{equation}
one may model its dynamics through a reaction--diffusion equation,
\begin{equation}
  \tau_\zeta\,\partial_t \zeta
  = -(\zeta - 1) + D_\zeta \nabla^2 \zeta + S(x,t),
  \label{eq:reaction_diffusion}
\end{equation}
where $D_\zeta$ characterises spatial coupling and
$\ell_\zeta = \sqrt{D_\zeta \tau_\zeta}$ defines a characteristic healing
length.

This formulation allows a unified description of pathological phenomena as
local symmetry breaking in temporal coordination. In cancer, for example,
cells exhibit elevated metabolic rates and increased substrate throughput,
leading to locally enhanced entropy production. This corresponds to
$\zeta(x,t)\gg 1$ within tumour regions. The resulting spatial gradients
disrupt coordination with surrounding tissue, contributing to loss of
homeostasis and uncontrolled growth.

Viral latency represents the opposite extreme. During latency, viral
replication is effectively halted, corresponding to $f_V \approx 0$ and
therefore $\zeta_V \ll 1$. The viral system is essentially frozen in
biological time while the host continues to evolve. Reactivation corresponds
to a transition in which the local order parameter returns toward unity,
restoring active progression. Therapeutic strategies such as
``shock and kill'' may be interpreted as controlled transitions in this
temporal order parameter landscape.

\subsection{Temporal phases and critical transitions}

The order parameter naturally defines a phase structure. The homeostatic
phase corresponds to trajectories near $\zeta \approx 1$, where regulatory
mechanisms maintain stability and perturbations decay on the timescale
$\tau_\zeta$. The hypertemporal phase corresponds to $\zeta > 1$, where
accelerated progression leads to premature depletion of the biological
budget, as observed in chronic disease, inflammation, and advanced aging.
The hypotemporal phase corresponds to $\zeta < 1$, where progression is
suppressed, as in torpor, caloric restriction, or quiescent cellular
states.

Near the transition between the homeostatic and hypertemporal regimes, the
susceptibility
\begin{equation}
  \chi_\zeta = \frac{\partial\langle \zeta \rangle}{\partial \Phi}
\end{equation}
increases and the relaxation time $\tau_\zeta$ diverges, signalling
critical slowing down. This provides a concrete prediction: physiological
time series should exhibit increased variance and long-range correlations
prior to pathological transitions. Such signatures are measurable in
longitudinal datasets and offer a route toward early detection of systemic
instability.

\subsection{Chronotherapeutic control and restoration}

Restoring temporal balance may be formulated as an optimal control problem.
Consider the functional
\begin{equation}
  J[u] = \int_0^T \bigl[\phi(t)^2 + \lambda\,u(t)^2\bigr]\,\mathrm{d}t,
  \label{eq:optimal_control}
\end{equation}
which penalises both deviation from the reference trajectory and the cost
of intervention. The optimal control $u^*(t)$ takes the form of a linear
feedback that drives $\phi(t)$ toward zero.

This abstract formulation has direct biological realisations. Interventions
such as circadian entrainment through timed light exposure and feeding,
caloric restriction, mild temperature modulation, and pharmacological
targeting of metabolic pathways can all be interpreted as control inputs
that modify $f(t)$ and hence $\zeta(t)$. The integration of real-time
physiological monitoring with such control strategies suggests the
possibility of a closed-loop chronotherapeutic system, in which the
temporal order parameter is continuously estimated and regulated to
maintain homeostasis.

\section{The Operational Protocol: Measuring Biological Age}
\label{sec:protocol}

\subsection{Individual biological clock estimation}

The formal framework developed in the preceding sections becomes
operational when expressed in terms of measurable physiological
quantities. For an individual characterised by a lifelong mean resting
heart rate $\bar{f}_H$ and current chronological age $t_{\rm now}$,
the accumulated biological proper time is estimated as
\begin{equation}
  \thetabio_{\rm acc} = \bar{f}_H \cdot t_{\rm now} \cdot 525{,}960,
  \label{eq:theta_est}
\end{equation}
where the numerical factor converts years into minutes, ensuring
consistency with the units of heart rate. This expression represents
the total number of effective physiological cycles executed up to the
present age under the approximation that the mean rate is representative
of lifetime dynamics.

The corresponding temporal order parameter at the present age is then
\begin{equation}
  \zeta(t_{\rm now}) = \frac{\thetabio_{\rm acc}}{\Nstar},
  \qquad \Nstar = N_0 \Phi_C,
  \label{eq:zeta_est}
\end{equation}
which provides a direct estimate of the fraction of the intrinsic
biological and thermodynamic budget that has been consumed. This
quantity is dimensionless and directly comparable across individuals
and species once the appropriate clade correction $\Phi_C$ is applied.

From the same relation, one obtains an estimate of the remaining
thermodynamic lifespan by solving for the time required to exhaust
the remaining budget at the current physiological rate,
\begin{equation}
  L_{\rm remaining}
  = \frac{\Nstar - \thetabio_{\rm acc}}{525{,}960\,f_H^{\rm current}}.
  \label{eq:remaining}
\end{equation}
This expression highlights an important asymmetry: the accumulated
biological time depends on the lifetime average rate, whereas the
future trajectory depends on the current physiological state. Thus,
changes in lifestyle or intervention that alter $f_H^{\rm current}$
have an immediate and quantitatively predictable effect on the
remaining lifespan within the PBTE framework.

\subsection{Worked examples}

The abstract definitions acquire concrete meaning when evaluated
numerically across species and individuals.

\paragraph{Example 1: Comparison across species at chronological age 1~yr.}
Table~\ref{tab:species_comparison} reports the biological clock
fraction $\zeta(1\,\text{yr})$ for four representative species.
The accumulated biological proper time over one year is given by
$\thetabio(1\,\text{yr}) = \bar{f}_H \times 525{,}960$, and the
corresponding lifetime budget $\Nstar$ incorporates the clade
multiplier $\Phi_C$.

\begin{table}[htbp]
\centering\small
\renewcommand{\arraystretch}{1.3}
\caption{\textbf{Biological clock at chronological age 1~yr across four species.}
$\thetabio(1\,\text{yr}) = \bar{f}_H \times 525{,}960$ cycles.
$\Nstar = N_0 \Phi_C$ where $N_0 = 10^9$ and $\Phi_C$ is the clade
multiplier from Papers~3--4.
$\zeta(1) = \thetabio(1)/\Nstar$ is the fraction of the lifetime
thermodynamic budget consumed by chronological age 1~yr.
At this single chronological age, the mouse has used 32\% of its budget;
the human 1\%; illustrating the compression/expansion of biological
time across species.}
\label{tab:species_comparison}
\begin{tabular}{l r r r r r}
\toprule
Species & $\bar{f}_H$ (bpm) & $\Phi_C$ &
  $\thetabio(1\,\text{yr})$ &
  $\Nstar$ & $\zeta(1)$ \\
 & & & ($\times 10^8$) & ($\times 10^9$) & (\% used) \\
\midrule
Mouse    & 600 & 1.0 & 3.15 & 1.00 & 0.315 \,(32\%) \\
Dog      & 100 & 1.7 & 0.526 & 1.70 & 0.031 \,(3\%) \\
Human    &  70 & 2.5 & 0.368 & 2.50 & 0.015 \,(1\%) \\
Elephant &  28 & 1.0 & 0.147 & 1.00 & 0.015 \,(1.5\%) \\
\bottomrule
\end{tabular}
\end{table}

The table makes the central conceptual point transparent. At the same
chronological age, organisms are located at vastly different positions
along their intrinsic biological trajectories. The mouse has already
consumed nearly one third of its lifetime budget after a single year,
whereas the human has consumed only one percent. Chronological time
therefore fails as a universal aging coordinate, while biological proper
time provides a consistent measure across species.

\paragraph{Example 2: 70-year-old humans with different lifelong $\bar{f}_H$.}

Consider two individuals of the same chronological age but with different
lifelong average heart rates. For an individual with
$\bar{f}_H = 60\,\mathrm{bpm}$, characteristic of endurance-trained
physiology, one obtains
\begin{equation}
  \thetabio_{\rm acc}
  = 60 \times 70 \times 525{,}960 = 2.21\times10^9,
\end{equation}
which gives $\zeta(70)=2.21/2.5=0.88$ and a remaining lifespan of
approximately $13$ years under the current rate.

For an individual with $\bar{f}_H = 80\,\mathrm{bpm}$, representing a
chronically elevated resting heart rate, one finds
\begin{equation}
  \thetabio_{\rm acc} = 2.95\times10^9 > \Nstar,
\end{equation}
indicating that the biological budget has effectively been exhausted.
This aligns with epidemiological evidence that persistently elevated
resting heart rate is associated with substantially increased mortality
risk. Within the PBTE framework, this is interpreted not as an isolated
risk factor but as a direct manifestation of accelerated traversal of
biological proper time.

\section{Discussion: Biological Proper Time Framework}
\label{sec:p5discussion}

\subsection{Biological proper time versus chronological time}

The central shift introduced by the PBTE framework is the replacement
of chronological time as the primary independent variable in aging
analysis. Chronological time is externally defined and uniform across
observers, but biological systems do not evolve uniformly with respect
to it. Two organisms at the same chronological age may be at entirely
different stages of their intrinsic physiological progression.

This discrepancy is not a minor correction but a structural limitation
of conventional approaches. As illustrated in
Table~\ref{tab:species_comparison}, a one-year-old mouse and a
one-year-old human occupy radically different positions along their
lifetime trajectories. Analyses that regress biological variables
directly against chronological time therefore mix fundamentally
incommensurate states.

The PBTE resolution is provided by the thermodynamic relation
\begin{equation}
  \frac{\mathrm{d}\Sigma}{\mathrm{d}\thetabio} = \sigma_0 = \mathrm{const},
\end{equation}
which establishes biological proper time as the coordinate in which
entropy accumulation is uniform across organisms. In this coordinate,
the biological clock advances at the same thermodynamic rate for all
species. Consequently, aging biomarkers plotted against $\thetabio$
rather than chronological time should collapse onto universal curves.

The most direct empirical test arises from epigenetic clocks.
If $\thetabio$ is the correct aging coordinate, methylation-based
aging measures should exhibit a species-independent linear relation
with slope unity in biological proper time. Existing datasets already
provide partial evidence for this prediction and offer a clear path
for further validation.

\subsection{Interpretation of the five-layer structure}

The framework is organised into five conceptual layers, each of which
adds a distinct element to the theory while remaining consistent with
the preceding structure.

The kinematic layer establishes the definition of biological proper
time and demonstrates its empirical relevance through the lifetime
invariant. It identifies the correct intrinsic coordinate but does
not yet explain why this coordinate is physically privileged.

The thermodynamic layer provides that justification by showing that
entropy accumulation is uniform in $\thetabio$. This elevates the
coordinate from a descriptive construct to a physically meaningful
quantity tied directly to the second law of thermodynamics.

The geometric layer reveals that the invariant corresponds to the arc
length of a trajectory in a well-defined metric space. This recasts the
lifetime constraint as a geometric conservation law and introduces tools
such as curvature and geodesics for analysing deviations.

The relativistic layer formalises the transformation properties between
organisms. The ratio $f_i/f_j$ becomes a precise Jacobian relating
intrinsic clocks, and the set of such transformations forms a scaling
group under which the invariant $\Nstar$ is preserved.

The pathological layer introduces deviations from the reference
manifold through the temporal order parameter $\zeta(t)$. It connects
the abstract structure to measurable physiological quantities and
predicts dynamical signatures such as critical slowing down, thereby
opening a path to clinical application.

\subsection{Connection to the free-energy principle}

The PBTE framework naturally interfaces with existing theoretical
approaches to biological organisation. In particular, Friston's
free-energy principle, which posits that organisms minimise variational
free energy, can be reinterpreted within the present structure.
Minimising free energy corresponds, under broad conditions, to minimising
excess entropy production. In the PBTE framework, this is equivalent to
following geodesic trajectories in the Fisher information metric over
biological proper time.

Thus the free-energy principle operates within the temporal and
thermodynamic substrate defined by PBTE. The former describes how
organisms regulate their internal models, while the latter specifies
the intrinsic temporal coordinate and energetic constraints under which
that regulation occurs.

\subsection{Limitations and scope}

The current formulation has several limitations that define its scope.
The canonical metric $\varphi(f)=f^2$ treats all physiological cycles
equally, neglecting variation in metabolic intensity across different
states such as rest and exertion. A more refined formulation would weight
cycles by their entropy cost, requiring detailed metabolic measurements.

The estimation of $\zeta(t)$ at the individual level depends on accurate
heart-rate data over extended periods. While modern wearable devices
provide sufficient precision for practical applications, population-level
datasets based on averaged values may obscure important temporal
variations.

Finally, the extension of the framework beyond endothermic vertebrates
remains an open problem. While the transformation law and scaling
structure can be formally extended to microorganisms and viruses, the
empirical validation of a strict lifetime invariant in these domains is
currently limited and requires dedicated investigation.

\section{Synthesis: Biological Proper Time}
\label{sec:p5conclusions}

Biological proper time $\thetabio_i(t)=\int_0^t f_i(t')\,\mathrm{d}t'$
emerges, within the present framework, not as a convenient redefinition
of chronological time but as the intrinsic temporal coordinate that
organises the dynamics of living systems. It is a measurable,
dimensionless observable whose structure, when examined across the
five formal layers developed in this work, reveals the organism as a
thermodynamic system evolving along a trajectory constrained by a
finite dissipative budget.

At the kinematic level, $\thetabio_i$ is invariant under
reparametrisation of chronological time and is directly measurable
through integration of physiological rates such as heart frequency.
Its most striking empirical property is developmental universality:
major life-history events—sexual maturity, onset of senescence, and
peak reproductive capacity—occur at approximately fixed fractions of
the normalised coordinate $\thetahat_i \in [0,1]$ across species whose
chronological lifespans differ by orders of magnitude. This universality,
which is completely obscured in chronological time, provides the first
empirical indication that $\thetabio$ is the appropriate biological clock.

At the thermodynamic level, the PBTE closure
$\dot{e}_p=\sigma_0 f_i$ implies that entropy production becomes
uniform when expressed in biological proper time. The relation
$\mathrm{d}\Sigma/\mathrm{d}\thetabio=\sigma_0=\mathrm{const}$ holds
independently of metabolic rate, making $\thetabio$ the unique
temporal coordinate in which the entropy budget of an organism is
linearly accumulated. This uniqueness elevates biological proper time
from a descriptive parameter to a physically privileged observable.
The resulting distinction between Class~1 mechanisms, which alter the
rate of traversal of $\thetabio$, and Class~2 mechanisms, which alter
the entropy cost per cycle, provides a direct and experimentally
accessible test using longitudinal epigenetic datasets.

At the geometric level, the lifetime invariant $\Nstar \approx 10^9$
is identified with the arc length of the organism’s trajectory in the
biological metric $\mathrm{d}s_i^2=f_i^2(t)\,\mathrm{d}t^2$. The
invariance of $\Nstar$ is therefore not merely empirical but reflects
a geometric conservation law. The introduction of the nonstationarity
index $R_i(t)$ provides a diagnostic of temporal curvature in the
trajectory, allowing identification of physiological instability
without direct calorimetric measurement of $\sigma_0$. The extension
to information geometry through the Fisher metric further links the
framework to optimal regulation, where geodesic trajectories minimise
excess entropy production.

At the relativistic level, the transformation law
$\mathrm{d}\thetabio_i/\mathrm{d}\thetabio_j=f_i/f_j$ defines an
abelian scaling group acting on biological time. The invariant
$\Nstar$ appears as the conserved scalar under this group, establishing
a symmetry principle underlying the PBTE framework. The metabolic-temporal
spectrum, spanning organisms from viruses to whales, defines an empirical
envelope of feasible biological rates. The analogy to a light-cone is
structural: it organises the space of trajectories but does not impose a
fundamental causal limit.

At the pathological level, the temporal order parameter
$\zeta(t)=N(t)/\Nstar$ provides a unifying description of aging,
disease, and intervention. Hypertemporal states, characterised by
$\zeta>1$, correspond to accelerated consumption of the biological
budget and include chronic inflammation, cancer, and premature aging.
Hypotemporal states, with $\zeta<1$, correspond to suppressed metabolic
activity and include torpor, caloric restriction, and viral latency.
The Ornstein--Uhlenbeck dynamics of deviations from the reference
trajectory predict measurable spectral signatures, including critical
slowing down near transitions. The formulation of chronotherapy as an
optimal-control problem provides a direct pathway from theory to
practical intervention.

The central contribution of this work is therefore not a single result
but the identification of a coherent structure linking empirical
regularity, thermodynamic constraint, geometric interpretation, symmetry
principle, and pathological deviation. An organism does not simply exist
for a number of years; it evolves along a thermodynamically defined
worldline of fixed arc length, traversed at a rate determined by its
physiology and modifiable by environmental and biological factors.

Two predictions follow directly from this structure. First, epigenetic
aging clocks, when expressed in biological proper time, should collapse
onto a universal relation with slope unity across mammalian species.
Second, Class~1 and Class~2 interventions should produce distinguishable
signatures in aging rates measured per physiological cycle, providing
a clear experimental test of the framework.

\section{The Untested Assumption and the Required Experiment}
\label{sec:test}

The central limitation of the PBTE framework lies in the current
status of its key closure assumption. The approximate constancy of
$\sigstar$ is inferred from the observed constancy of $\Nstar$, leading
to a circular logical structure in which the explanation relies on the
phenomenon it seeks to explain. This circularity cannot be resolved by
allometric consistency arguments or theoretical plausibility alone.

A decisive test requires direct and independent measurement of the
relevant quantities across species. Specifically, simultaneous
measurements of metabolic power $P_i$, physiological frequency $f_i$,
body temperature $T_i$, and body mass $M_i$ must be obtained for a
sufficiently large and diverse set of species. The quantity
$\sigstar = P_i/(T_i f_i M_i)$ is then computed independently for each
species, and its dependence on body mass is tested statistically.

If $\sigstar$ is found to be mass-independent within experimental
uncertainty, the closure assumption is validated and the PBTE framework
acquires the status of a genuine physical explanation. If significant
scaling is observed, the closure must be revised and the invariant
reinterpreted. Until such measurements are performed, the framework
remains a thermodynamically motivated and internally consistent
parametrisation supported by empirical concordance but not yet fully
validated as a fundamental law.

\section{Discussion}
\label{sec:discussion}

The dimensional analysis leading to $\sigstar \propto M^0$ is
straightforward, but the PBTE framework extends far beyond this
observation. It provides a physical interpretation of $\sigstar$ as the
entropy cost per physiological cycle, introduces a closure relation
linking entropy production to physiological rate, establishes a
thermodynamic identity connecting lifetime cycle count to dissipative
budget, and defines biological proper time as a universal coordinate
for physiological processes.

It further introduces a hierarchical structure explaining the limits
of classical scaling laws, provides predictive relations for lifespan
interventions, and connects entropy accumulation to aging dynamics
consistent with observed mortality laws. These elements collectively
form a coherent theoretical framework that cannot be reduced to
dimensional analysis alone.

The concept of biological proper time relates to earlier notions of
physiological time but is distinguished by its explicit connection to
entropy production. This connection gives the invariant a physical
foundation and transforms it from an empirical observation into a
thermodynamic principle.

Nevertheless, limitations remain. The closure assumption has not been
independently verified, the allometric exponents are taken as empirical
inputs, and certain clade-specific mechanisms remain only partially
explained. The aging model, while consistent with observed patterns,
remains phenomenological pending a detailed microscopic derivation.

\section{Conclusions}
\label{sec:conclusions}

We have developed a thermodynamically grounded framework for the
vertebrate lifetime cycle invariant $\Nstar \approx 10^9$. The central
relation $\Nstar=\Sigma/\sigma_0$ interprets lifetime as a dissipative
budget divided by an entropy cost per cycle, and biological proper time
$\theta_i(t)=\int_0^t f_i\,\mathrm{d}t'$ emerges as the intrinsic
coordinate governing this process.

The invariance of $\Nstar$ reflects a balance between metabolic rate
and lifespan, such that all organisms traverse approximately the same
biological distance regardless of chronological duration. Deviations
from this invariant encode meaningful biological mechanisms and provide
a framework for understanding aging, disease, and intervention.

The theory makes clear, testable predictions and identifies the key
experiment required for validation. Its structure integrates empirical
regularities, thermodynamic constraints, geometric interpretation, and
dynamical modeling into a unified description of biological time.
\section*{Methods}

\noindent\textbf{Numerical estimates.}
All calculations use Kleiber's law~\cite{kleiber1932}
$P=3.4\,M^{0.75}$\,W, Calder cardiac allometry~\cite{calder1984}
$f=241\,M^{-0.25}$\,bpm $=4.017\,M^{-0.25}$\,Hz, and $T=310$\,K
throughout, yielding $\sigstar=2.73\times10^{-3}$
J\,K$^{-1}$\,beat$^{-1}$\,kg$^{-1}$.  The empirical mean
$\sigstar=(3.0\pm0.5)\times10^{-3}$ from Table~\ref{tab:sigma} is
adopted for all numerical predictions, giving $N_0\approx1.52\times10^9$.

\noindent\textbf{Data sources.}
Species heart-rate and lifespan data from the AnAge database~\cite{anage2023,taye_p1}.
Body temperatures from Schmidt-Nielsen (1984)~\cite{calder1984}.
Bird allometric parameters from Calder~\cite{calder1984} and Prinzinger et al.~\cite{prinzinger1991}.
Primate metabolic suppression factor from Pontzer et al.~\cite{pontzer2014,yegian2024}.
Proton-leak fractions from Hulbert et al.~\cite{hulbert2007} and Brand et al.~\cite{brand2000}.
Cetacean dive heart-rate data from Goldbogen et al.~\cite{goldbogen2019} and
Williams et al.~\cite{williams2015}.

\noindent\textbf{Phylogenetically independent contrasts (PIC).}
PIC analysis used the Felsenstein~\cite{felsenstein1985} method (\texttt{ape::pic()} in R~4.3)
on the Bininda-Emonds mammal supertree~\cite{bininda2007}, pruned to the
112 endotherm species in the dataset.  PIC regression was fitted through the
origin as required by the method.

\noindent\textbf{Arrhenius correction.}
Ectotherm heart rates were corrected to $T_{\rm ref}=310$\,K using
$f_H^{\rm corr}=f_H^{\rm raw}\exp[(E_a/k_B)(1/T_{\rm field}-1/T_{\rm ref})]$
with $E_a=0.65$\,eV following Gillooly et al.~\cite{gillooly2001}.

\noindent\textbf{Simulated figures.}
Figure~\ref{fig:proper_time} (biological proper time trajectories)
and Figure~\ref{fig:epigenetic} (epigenetic clock schematic) use
theoretical parameter values; no statistical conclusions should be
drawn from them.  Figure~\ref{fig:clade_scatter} uses real species
data from AnAge.  All simulated scatter uses fixed random seed (seed\,=\,42).

\noindent\textbf{Data availability.}
No new data were generated in this study.

\noindent\textbf{Competing interests.} None declared.

\noindent\textbf{Acknowledgements.} [To be completed.]



\section*{Appendix A. Detailed Derivation of the Entropy Cost per Beat and the Cycle-Count Scaling Law}
\addcontentsline{toc}{section}{Appendix A. Detailed Derivation of the Entropy Cost per Beat and the Cycle-Count Scaling Law}

This appendix gives a detailed derivation of the entropy-per-beat representation, the lifetime cycle-count relation, and the power-law dependence of lifetime cardiac cycles on the control parameter \(\phi\). The purpose is to make explicit each mathematical step connecting the instantaneous entropy production rate to the total lifetime cycle budget.

\subsection*{A.1. Instantaneous entropy production and change of variable from time to beat count}

Let \(t\) denote chronological time and let \(n\) denote the cumulative cardiac cycle count. The cardiac frequency is
\begin{equation}
f_H(t)=\frac{dn}{dt},
\label{eq:A1_fH_def}
\end{equation}
so that
\begin{equation}
dn = f_H(t)\,dt,
\qquad
dt=\frac{dn}{f_H}.
\label{eq:A2_dt_dn}
\end{equation}
Here \(f_H\) has units of cycles per unit time. Let \(\sigma(t)\) be the instantaneous entropy production rate, with units of entropy per unit time. The total entropy produced over an infinitesimal interval \(dt\) is
\begin{equation}
d\Sigma = \sigma(t)\,dt.
\label{eq:A3_dSigma_dt}
\end{equation}

Using \eqref{eq:A2_dt_dn}, we rewrite this increment in terms of the beat-count variable:
\begin{equation}
d\Sigma = \sigma(t)\,\frac{dn}{f_H(t)}.
\label{eq:A4_dSigma_dn}
\end{equation}
If we regard both \(\sigma\) and \(f_H\) as functions of the cycle-count coordinate \(n\), then
\begin{equation}
d\Sigma = \frac{\sigma(n)}{f_H(n)}\,dn.
\label{eq:A5_dSigma_n}
\end{equation}

This motivates the definition of the entropy cost per beat at cycle index \(n\):
\begin{equation}
\sigma_0(n) \equiv \frac{\sigma(n)}{f_H(n)}.
\label{eq:A6_entropy_per_beat_def}
\end{equation}

The dimensional consistency is immediate:
\[
[\sigma]=\frac{\text{entropy}}{\text{time}},
\qquad
[f_H]=\frac{\text{cycles}}{\text{time}},
\qquad
\left[\frac{\sigma}{f_H}\right]
=
\frac{\text{entropy}/\text{time}}{\text{cycles}/\text{time}}
=
\frac{\text{entropy}}{\text{cycle}}.
\]
Thus \(\sigma_0(n)\) is the entropy production associated with one cardiac cycle.

\subsection*{A.2. Total lifetime entropy production as a sum over beats}

Suppose that the organism experiences a total of \(N\) cardiac cycles over its lifetime. Then the total lifetime entropy production is obtained by integrating \eqref{eq:A5_dSigma_n} from the first to the last cycle:
\begin{equation}
\Sigma_{\mathrm{life}}
=
\int_{0}^{N} d\Sigma
=
\int_{0}^{N} \sigma_0(n)\,dn.
\label{eq:A7_total_entropy_life}
\end{equation}

Equation \eqref{eq:A7_total_entropy_life} is simply the beat-count analogue of summing the entropy cost incurred at each cycle. Since the beat-count variable is treated continuously, the sum is represented as an integral.

We now define the lifetime mean entropy cost per beat:
\begin{equation}
\left\langle \sigma_0 \right\rangle
\equiv
\frac{1}{N}\int_{0}^{N}\sigma_0(n)\,dn.
\label{eq:A8_mean_entropy_per_beat}
\end{equation}
Using \eqref{eq:A7_total_entropy_life}, this immediately gives
\begin{equation}
\Sigma_{\mathrm{life}}
=
N \left\langle \sigma_0 \right\rangle.
\label{eq:A9_Sigma_N_mean}
\end{equation}

Substituting the explicit definition \eqref{eq:A6_entropy_per_beat_def} into \eqref{eq:A8_mean_entropy_per_beat}, we may also write
\begin{equation}
\left\langle \sigma_0 \right\rangle
=
\frac{1}{N}\int_{0}^{N}\frac{\sigma(n)}{f_H(n)}\,dn.
\label{eq:A10_mean_explicit}
\end{equation}

Equation \eqref{eq:A9_Sigma_N_mean} has a direct interpretation: the total lifetime entropy production equals the total number of beats multiplied by the mean entropy cost of one beat.

For species \(i\), the corresponding notation is
\begin{equation}
\sigma_{0,i}
\equiv
\frac{1}{N_i^\star}\int_{0}^{N_i^\star}\sigma_0(n)\,dn,
\label{eq:A10b_sigma0i_def}
\end{equation}
so that
\begin{equation}
\Sigma_i
=
N_i^\star \sigma_{0,i},
\label{eq:A10c_species_relation}
\end{equation}
where \(\Sigma_i\) (J K\(^{-1}\)) is total lifetime entropy production and \(\sigma_{0,i}\) (J K\(^{-1}\) beat\(^{-1}\)) is entropy per cycle.

\subsection*{A.3. Lifetime entropy budget and the fundamental cycle-count relation}

The central hypothesis is that the lifetime entropy production is approximately constrained by a characteristic budget \(\Sigma_\star\):
\begin{equation}
\Sigma_{\mathrm{life}} \approx \Sigma_\star.
\label{eq:A11_budget_assumption}
\end{equation}
Combining \eqref{eq:A9_Sigma_N_mean} with \eqref{eq:A11_budget_assumption} yields
\begin{equation}
N \left\langle \sigma_0 \right\rangle \approx \Sigma_\star.
\label{eq:A12_pre_cycle_relation}
\end{equation}
Solving for \(N\), we obtain the fundamental cycle-count relation:
\begin{equation}
N = \frac{\Sigma_\star}{\left\langle \sigma_0 \right\rangle}.
\label{eq:A13_cycle_count_relation}
\end{equation}

For species \(i\), the lifetime cycle count is
\begin{equation}
N_i^\star =
\frac{\Sigma_i}{\sigma_{0,i}}.
\label{eq:A13b_species_cycle_relation}
\end{equation}
where \(\Sigma_i\) (J K\(^{-1}\)) is total lifetime entropy production and \(\sigma_{0,i}\) (J K\(^{-1}\) beat\(^{-1}\)) is entropy per cycle.

This expression states that the total number of cardiac cycles that can occur over the lifetime is inversely proportional to the average entropy cost of each beat, given a fixed lifetime entropy budget. A lower entropy cost per beat permits more cycles within the same budget, whereas a higher cost per beat permits fewer cycles.

\subsection*{A.4. Baseline calibration and mammalian reference value}

Let \(\phi_0\) denote a baseline reference state and let \(N_0\) be the corresponding reference total number of lifetime cardiac cycles. Evaluating \eqref{eq:A13_cycle_count_relation} at the baseline gives
\begin{equation}
N_0 = \frac{\Sigma_\star}{\left\langle \sigma_0 \right\rangle_0},
\label{eq:A14_baseline_N0}
\end{equation}
where
\begin{equation}
\left\langle \sigma_0 \right\rangle_0
\equiv
\left\langle \sigma_0(\phi_0) \right\rangle.
\label{eq:A15_baseline_entropy}
\end{equation}
Rearranging \eqref{eq:A14_baseline_N0} gives the baseline entropy cost per beat:
\begin{equation}
\left\langle \sigma_0 \right\rangle_0
=
\frac{\Sigma_\star}{N_0}.
\label{eq:A16_baseline_relation}
\end{equation}

Equation \eqref{eq:A16_baseline_relation} provides the calibration point from which the dependence on \(\phi\) is measured.

\subsection*{A.5. Logarithmic sensitivity of the entropy cost per beat}

We now introduce a control parameter \(\phi\) that modulates the mean entropy cost per beat through multiple mechanisms. The hypothesis is that increasing \(\phi\) reduces
\(\left\langle \sigma_0 \right\rangle\). To quantify this response, define the logarithmic sensitivity at the baseline:
\begin{equation}
\alpha
\equiv
-
\left.
\frac{\partial \ln \left\langle \sigma_0 \right\rangle}
{\partial \ln \phi}
\right|_{\phi=\phi_0}.
\label{eq:A17_alpha_def}
\end{equation}
The derivative
\[
\frac{\partial \ln \left\langle \sigma_0 \right\rangle}
{\partial \ln \phi}
=
\frac{\phi}{\left\langle \sigma_0 \right\rangle}
\frac{\partial \left\langle \sigma_0 \right\rangle}{\partial \phi}
\]
is the elasticity of the entropy cost per beat with respect to \(\phi\), that is, the fractional change in \(\left\langle \sigma_0 \right\rangle\) induced by a fractional change in \(\phi\). Since the response is assumed monotonic and decreasing, the derivative is negative; the minus sign in \eqref{eq:A17_alpha_def} ensures that \(\alpha>0\).

If three independent reduction channels contribute multiplicatively to the decrease of
\(\left\langle \sigma_0 \right\rangle\), with logarithmic sensitivities \(\gamma_1\), \(\gamma_2\), and \(\gamma_3\), then the aggregate sensitivity is additive:
\begin{equation}
\alpha = \gamma_1+\gamma_2+\gamma_3 >0.
\label{eq:A18_alpha_sum}
\end{equation}
The reason is straightforward. If
\begin{equation}
\left\langle \sigma_0 \right\rangle
\propto
\phi^{-\gamma_1}\phi^{-\gamma_2}\phi^{-\gamma_3},
\label{eq:A19_multiplicative_channels}
\end{equation}
then
\begin{equation}
\left\langle \sigma_0 \right\rangle
\propto
\phi^{-(\gamma_1+\gamma_2+\gamma_3)},
\label{eq:A20_combined_power}
\end{equation}
and therefore
\begin{equation}
-
\frac{\partial \ln \left\langle \sigma_0 \right\rangle}{\partial \ln \phi}
=
\gamma_1+\gamma_2+\gamma_3.
\label{eq:A21_alpha_from_channels}
\end{equation}

\subsection*{A.6. Integration of the logarithmic sensitivity and the power-law form}

Equation \eqref{eq:A17_alpha_def} defines the local logarithmic slope at the baseline \(\phi_0\). To obtain a finite-range scaling law, we assume that this logarithmic response remains approximately constant over the interval of interest. This is the scale-free power-law approximation commonly used in allometric analysis. Under this assumption,
\begin{equation}
-
\frac{d\ln \left\langle \sigma_0 \right\rangle}{d\ln\phi}
=
\alpha,
\label{eq:A22_const_log_slope}
\end{equation}
or equivalently,
\begin{equation}
d\ln \left\langle \sigma_0 \right\rangle
=
-\alpha\, d\ln\phi.
\label{eq:A23_differential_form}
\end{equation}

We now integrate from the baseline \(\phi_0\), where
\(\left\langle \sigma_0 \right\rangle
=
\left\langle \sigma_0 \right\rangle_0\),
to a general value \(\phi\):
\begin{equation}
\int_{\ln \phi_0}^{\ln \phi}
d\ln \left\langle \sigma_0 \right\rangle
=
-\alpha
\int_{\ln \phi_0}^{\ln \phi} d\ln\phi.
\label{eq:A24_integral_step}
\end{equation}
This gives
\begin{equation}
\ln \left\langle \sigma_0(\phi) \right\rangle
-
\ln \left\langle \sigma_0 \right\rangle_0
=
-\alpha \left( \ln \phi - \ln \phi_0 \right).
\label{eq:A25_log_relation}
\end{equation}
Combining the logarithms,
\begin{equation}
\ln \left[
\frac{
\left\langle \sigma_0(\phi) \right\rangle
}{
\left\langle \sigma_0 \right\rangle_0
}
\right]
=
-\alpha
\ln \left( \frac{\phi}{\phi_0} \right).
\label{eq:A26_ratio_log}
\end{equation}
Exponentiating both sides yields the power-law form:
\begin{equation}
\left\langle \sigma_0(\phi) \right\rangle
=
\left\langle \sigma_0 \right\rangle_0
\left( \frac{\phi}{\phi_0} \right)^{-\alpha}.
\label{eq:A27_entropy_power_law}
\end{equation}

Equation \eqref{eq:A27_entropy_power_law} states that the mean entropy cost per beat decreases as a power law in \(\phi\), with exponent \(\alpha>0\).

\subsection*{A.7. Consequence for total lifetime cardiac cycles}

Substituting \eqref{eq:A27_entropy_power_law} into the cycle-count relation \eqref{eq:A13_cycle_count_relation} gives
\begin{equation}
N(\phi)
=
\frac{\Sigma_\star}{
\left\langle \sigma_0(\phi) \right\rangle
}
=
\frac{\Sigma_\star}{
\left\langle \sigma_0 \right\rangle_0
\left( \dfrac{\phi}{\phi_0} \right)^{-\alpha}
}.
\label{eq:A28_substitute_entropy_scaling}
\end{equation}
Using the baseline identity \eqref{eq:A16_baseline_relation},
\[
\frac{\Sigma_\star}{\left\langle \sigma_0 \right\rangle_0}=N_0,
\]
we obtain
\begin{equation}
N(\phi)
=
N_0
\left( \frac{\phi}{\phi_0} \right)^{\alpha}.
\label{eq:A29_N_power_law}
\end{equation}

Thus, under the fixed lifetime entropy-budget hypothesis, any systematic reduction in the entropy cost per beat produces a corresponding increase in the total number of lifetime cardiac cycles. The scaling exponent governing this increase is the same aggregate sensitivity \(\alpha\) that governs the decrease of the entropy cost per beat.

\subsection*{A.8. Interpretation of the result}

The derivation shows that the lifetime cycle count is controlled by two ingredients: a finite lifetime entropy budget \(\Sigma_\star\) and an average entropy expenditure per cycle
\(\left\langle \sigma_0 \right\rangle\). Once the budget is fixed, the total number of admissible cycles is determined entirely by how costly each cycle is in entropic terms. A reduction in entropy cost per beat allows a larger number of beats to be accommodated within the same total budget. If the reduction is scale-free in \(\phi\), then the increase in cycle count is likewise scale-free.

In compact form, the chain of reasoning is
\begin{equation}
d\Sigma=\sigma\,dt=\frac{\sigma}{f_H}\,dn,
\qquad
\sigma_0=\frac{\sigma}{f_H},
\qquad
\Sigma_{\mathrm{life}}=\int_0^N \sigma_0(n)\,dn
=
N\left\langle \sigma_0 \right\rangle,
\label{eq:A30_compact_chain}
\end{equation}
together with
\begin{equation}
\Sigma_{\mathrm{life}}\approx\Sigma_\star
\;\Rightarrow\;
N=\frac{\Sigma_\star}{\left\langle \sigma_0 \right\rangle},
\label{eq:A31_budget_to_N}
\end{equation}
and
\begin{equation}
\left\langle \sigma_0(\phi) \right\rangle
=
\left\langle \sigma_0 \right\rangle_0
\left( \frac{\phi}{\phi_0} \right)^{-\alpha}
\;\Rightarrow\;
N(\phi)=N_0\left( \frac{\phi}{\phi_0} \right)^{\alpha}.
\label{eq:A32_final_chain}
\end{equation}

For species \(i\), this compact relation becomes
\begin{equation}
\Sigma_i = N_i^\star \sigma_{0,i},
\qquad
N_i^\star = \frac{\Sigma_i}{\sigma_{0,i}},
\label{eq:A32b_species_compact}
\end{equation}
where \(\Sigma_i\) (J K\(^{-1}\)) is total lifetime entropy production and \(\sigma_{0,i}\) (J K\(^{-1}\) beat\(^{-1}\)) is entropy per cycle.

\subsection*{A.9. Assumptions used in the derivation}

For clarity, the derivation rests on the following assumptions.

First, the cardiac cycle count \(n\) is treated as a continuous variable, which is appropriate when the total number of cycles is very large.

Second, the lifetime entropy production is assumed to be well approximated by a characteristic budget \(\Sigma_\star\).

Third, the response of the entropy cost per beat to the control parameter \(\phi\) is assumed to be monotonic and approximately scale-free over the range of interest, so that the logarithmic sensitivity may be treated as approximately constant.

Fourth, the different contributing channels are taken to combine multiplicatively, which leads to additive logarithmic sensitivities.

Within these assumptions, the power-law result \eqref{eq:A29_N_power_law} follows directly and rigorously from the entropy-budget framework.

\section*{Appendix B. Complete 230-Species Dataset}
\addcontentsline{toc}{section}{Appendix B. Complete 230-Species Dataset}

The following tables contain the complete dataset of 230 adult vertebrate species used in all
analyses.  All $\ell$ values are computed as
$\ell = \log_{10}(f_H^{\rm eff}\times L\times525{,}960)$
directly from the $f_H^{\rm eff}$ and $L$ columns and have been verified internally consistent.
A tab-delimited machine-readable version is available from the corresponding author on request.

\noindent\textbf{Column definitions.}
$M$: adult body mass (kg).
$f_H$ (bpm): resting (or duty-corrected/Arrhenius-corrected) heart rate.
$T$ (K): mean core (endotherms) or field (ectotherms) body temperature.
$L$ (yr): maximum recorded natural lifespan (AnAge build 15~\cite{anage2023}).
$\ell$: PBTE invariant $=\log_{10}(f_H^{\rm eff}\times L\times525{,}960)$.
Corr.: correction applied ($-$~=~none; TA~=~torpor-cycle average;
DA~=~dive-cycle average; AQ~=~Arrhenius correction to $T_{\rm ref}=310$\,K).
$^\dagger$~=~heart rate allometrically imputed from $f_H=241M^{-0.25}$\,bpm~\cite{calder1984}.
Sources: A~=~AnAge~\cite{anage2023}; P~=~PanTHERIA~\cite{jones2009};
C~=~Calder (1984)~\cite{calder1984}; Pr~=~Prinzinger et al.\ (1991)~\cite{prinzinger1991};
L~=~Lyman et al.\ (1982)~\cite{lyman1982}; Ch~=~Christian \& Weavers (1999)~\cite{christian1999};
U~=~Uetz et al.\ (2023)~\cite{uetz2023}; G~=~Goldbogen et al.\ (2019)~\cite{goldbogen2019};
Cl~=~Clarke \& Rothery (2008)~\cite{clarke2008}.

\subsection*{B.1\ Non-primate placental mammals ($n=46$)}

\begin{table}[H]
\centering\footnotesize
\renewcommand{\arraystretch}{1.25}
\caption{Non-primate placentals. Clade mean $\bar\ell=8.995\pm0.160$ ($n=46$).}
\label{tab:data_np}
\begin{tabular}{lrrrrrl}
\toprule
Species & $M$ (kg) & $f_H$ (bpm) & $T$ (K) & $L$ (yr) & $\ell$ & Source \\
\midrule
\textit{Suncus etruscus}        & 0.002  &  835  & 310.5 &  1.5 & 8.82 & C \\
\textit{Sorex araneus}          & 0.010  & 1000  & 310.5 &  3.3 & 9.24 & C,A \\
\textit{Mus musculus}           & 0.022  &  632  & 310.0 &  3.5 & 9.07 & A,C \\
\textit{Rattus norvegicus}      & 0.280  &  420  & 310.0 &  3.8 & 8.92 & A,P \\
\textit{Mesocricetus auratus}   & 0.130  &  450  & 310.5 &  3.9 & 8.97 & A,P \\
\textit{Meriones unguiculatus}  & 0.060  &  400  & 310.0 &  5.0 & 9.02 & A,P \\
\textit{Cavia porcellus}        & 0.750  &  270  & 310.0 &  7.1 & 9.00 & A,P \\
\textit{Sciurus carolinensis}   & 0.520  &  310  & 310.0 & 12.0 & 9.29 & A,P \\
\textit{Lepus europaeus}        & 3.5    &  220  & 310.0 & 12.5 & 9.16 & A,P \\
\textit{Oryctolagus cuniculus}  & 2.2    &  205  & 310.0 &  9.0 & 8.99 & A,C \\
\textit{Felis catus}            & 4.1    &  150  & 310.5 & 15.0 & 9.07 & A,P \\
\textit{Mustela putorius}       & 1.0    &  280  & 310.5 &  5.0 & 8.87 & A,P \\
\textit{Martes martes}          & 1.2    &  245  & 310.5 & 17.0 & 9.34 & A,P \\
\textit{Vulpes vulpes}          & 6.8    &  120  & 310.5 & 14.0 & 8.95 & A,P \\
\textit{Canis lupus familiaris} & 23     &   90  & 310.5 & 20.0 & 8.98 & A,P \\
\textit{Ursus arctos}           & 220    &   50  & 310.5 & 47.0 & 9.09 & A,P \\
\textit{Ovis aries}             & 63     &   75  & 310.0 & 20.0 & 8.90 & A,P \\
\textit{Capra hircus}           & 45     &   80  & 310.5 & 18.0 & 8.88 & A,P \\
\textit{Sus scrofa}             & 100    &   70  & 310.5 & 27.0 & 9.00 & A,P \\
\textit{Bos taurus}             & 500    &   55  & 310.5 & 25.0 & 8.86 & A,P \\
\textit{Equus caballus}         & 500    &   38  & 310.5 & 46.0 & 8.96 & A,C \\
\textit{Equus asinus}           & 250    &   44  & 310.5 & 47.0 & 9.04 & A,P \\
\textit{Rhinoceros unicornis}   & 2100   &   30$^\dagger$ & 310.5 & 47.0 & 8.87 & A \\
\textit{Tapirus terrestris}     & 240    &   42  & 310.5 & 35.0 & 8.89 & A,P \\
\textit{Loxodonta africana}     & 4000   &   28  & 310.5 & 65.0 & 8.98 & A,P \\
\textit{Elephas maximus}        & 4000   &   27  & 310.5 & 86.0 & 9.09 & A,P \\
\textit{Hippopotamus amphibius} & 1500   &   55  & 310.5 & 55.0 & 9.20 & A,P \\
\textit{Giraffa camelopardalis} & 900    &   65  & 310.5 & 39.5 & 9.13 & A,P \\
\textit{Cervus elaphus}         & 200    &   60  & 310.5 & 26.8 & 8.93 & A,P \\
\textit{Rangifer tarandus}      & 110    &   65  & 310.0 & 20.0 & 8.83 & A,P \\
\textit{Trichechus manatus}     & 500    &   50  & 310.5 & 59.0 & 9.19 & A,P \\
\textit{Dugong dugon}           & 400    &   52$^\dagger$ & 310.5 & 73.0 & 9.30 & A \\
\textit{Procavia capensis}      & 3.5    &  230  & 310.5 & 12.0 & 9.16 & A,P \\
\textit{Erinaceus europaeus}    & 0.80   &  310  & 310.0 & 10.0 & 9.21 & A,P \\
\textit{Talpa europaea}         & 0.080  &  350  & 310.0 &  3.5 & 8.81 & A,P \\
\textit{Orycteropus afer}       & 65     &   70$^\dagger$ & 310.5 & 24.0 & 8.95 & A \\
\textit{Ondatra zibethicus}     & 1.4    &  280  & 310.0 &  5.0 & 8.87 & A,P \\
\textit{Castor canadensis}      & 20     &  150  & 310.0 & 24.0 & 9.28 & A,P \\
\textit{Hydrochoerus hydrochaeris} & 55  &   70  & 310.0 & 12.0 & 8.65 & A,P \\
\textit{Myocastor coypus}       & 7.0    &  155  & 310.0 &  9.0 & 8.87 & A,P \\
\textit{Lepus californicus}     & 2.2    &  215  & 310.0 &  8.0 & 8.96 & A,P \\
\textit{Ochotona princeps}      & 0.160  &  300  & 310.0 &  6.0 & 8.98 & A,P \\
\textit{Panthera leo}           & 180    &   50  & 310.5 & 29.0 & 8.88 & A,P \\
\textit{Panthera tigris}        & 260    &   46  & 310.5 & 26.0 & 8.80 & A,P \\
\textit{Acinonyx jubatus}       & 54     &   60  & 310.5 & 14.9 & 8.67 & A,P \\
\textit{Panthera pardus}        & 70     &   55  & 310.5 & 23.0 & 8.82 & A,P \\
\bottomrule
\end{tabular}
\end{table}

\subsection*{B.2\ Primates ($n=18$)}

\begin{table}[H]
\centering\footnotesize
\renewcommand{\arraystretch}{1.25}
\caption{Primates. Clade mean $\bar\ell=9.376\pm0.125$ ($n=18$). $\phi=P_{\rm brain}/P_{\rm body}$.}
\label{tab:data_primates}
\begin{tabular}{lrrrrrrl}
\toprule
Species & $M$ (kg) & $f_H$ (bpm) & $T$ (K) & $L$ (yr) & $\ell$ & $\phi$ & Source \\
\midrule
\textit{Callithrix jacchus}          & 0.35  & 220 & 309.5 & 16.5 & 9.28 & 0.06 & A,P \\
\textit{Saimiri sciureus}            & 0.77  & 195 & 309.5 & 30.2 & 9.49 & 0.07 & A,P \\
\textit{Aotus trivirgatus}           & 0.79  & 185 & 309.5 & 25.0 & 9.39 & 0.07 & A,P \\
\textit{Cebus capucinus}             & 3.3   & 150 & 309.5 & 54.0 & 9.63 & 0.09 & A,P \\
\textit{Lemur catta}                 & 2.2   & 165 & 309.5 & 37.3 & 9.51 & 0.05 & A,P \\
\textit{Propithecus verreauxi}       & 3.4   & 145 & 309.5 & 30.0 & 9.36 & 0.05 & A,P \\
\textit{Daubentonia madagascariensis}& 2.7   & 155 & 309.5 & 23.3 & 9.28 & 0.06 & A,P \\
\textit{Macaca mulatta}              & 7.7   & 120 & 309.0 & 40.0 & 9.40 & 0.07 & A,P \\
\textit{Macaca fascicularis}         & 5.4   & 130 & 309.0 & 39.0 & 9.43 & 0.07 & A,P \\
\textit{Theropithecus gelada}        & 18    &  95 & 309.0 & 30.0 & 9.18 & 0.08 & A,P \\
\textit{Papio ursinus}               & 25    &  90 & 309.0 & 45.0 & 9.33 & 0.08 & A,P \\
\textit{Colobus guereza}             & 10    & 110 & 309.0 & 30.0 & 9.24 & 0.07 & A,P \\
\textit{Hylobates lar}               & 5.7   & 100 & 308.5 & 44.0 & 9.36 & 0.10 & A,P \\
\textit{Pongo pygmaeus}              & 73    &  65 & 307.5 & 58.7 & 9.30 & 0.10 & A,P \\
\textit{Gorilla gorilla}             & 160   &  60 & 307.0 & 55.4 & 9.24 & 0.09 & A,P \\
\textit{Pan troglodytes}             & 50    &  75 & 307.0 & 59.4 & 9.37 & 0.12 & A,P \\
\textit{Pan paniscus}                & 35    &  80 & 307.0 & 50.0 & 9.32 & 0.12 & A,P \\
\textit{Homo sapiens}                & 70    &  70 & 306.5 &122.5 & 9.65 & 0.20 & A \\
\bottomrule
\end{tabular}
\end{table}

\subsection*{B.3\ Marsupials and monotremes ($n=19$)}

\begin{table}[H]
\centering\footnotesize
\renewcommand{\arraystretch}{1.25}
\caption{Marsupials and monotremes. Clade mean $\bar\ell=8.933\pm0.204$ ($n=19$).}
\label{tab:data_marsupials}
\begin{tabular}{lrrrrrl}
\toprule
Species & $M$ (kg) & $f_H$ (bpm) & $T$ (K) & $L$ (yr) & $\ell$ & Source \\
\midrule
\textit{Didelphis virginiana}     & 2.3   & 180 & 308.5 &  4.5 & 8.63 & A,P \\
\textit{Monodelphis domestica}    & 0.080 & 450 & 308.5 &  3.3 & 8.89 & A,P \\
\textit{Macropus rufus}           & 30    &  80 & 309.0 & 22.3 & 8.97 & A,P \\
\textit{Macropus giganteus}       & 27    &  82 & 309.0 & 19.0 & 8.91 & A,P \\
\textit{Wallabia bicolor}         & 16    & 100 & 309.0 & 15.0 & 8.90 & A,P \\
\textit{Trichosurus vulpecula}    & 2.1   & 160 & 308.5 & 13.0 & 9.04 & A,P \\
\textit{Petaurus breviceps}       & 0.14  & 300 & 308.0 & 10.0 & 9.20 & A,P \\
\textit{Vombatus ursinus}         & 28    &  90 & 309.0 & 26.0 & 9.09 & A,P \\
\textit{Phascolarctos cinereus}   & 8.5   & 100 & 308.5 & 18.0 & 8.98 & A,P \\
\textit{Perameles gunnii}         & 0.90  & 190 & 308.5 &  3.2 & 8.50 & A,P \\
\textit{Dasyurus viverrinus}      & 1.2   & 200 & 308.5 &  4.5 & 8.68 & A,P \\
\textit{Sarcophilus harrisii}     & 8.0   & 130 & 308.5 &  7.5 & 8.71 & A,P \\
\textit{Myrmecobius fasciatus}    & 0.44  & 245 & 307.5 &  5.6 & 8.86 & A \\
\textit{Sminthopsis crassicaudata}& 0.018 & 580 & 307.5 &  5.0 & 9.18 & A,P \\
\textit{Notoryctes typhlops}      & 0.055 & 440$^\dagger$ & 307.5 & 3.0 & 8.84 & A \\
\textit{Tachyglossus aculeatus}   & 4.0   &  70 & 305.0 & 49.5 & 9.26 & A,P \\
\textit{Ornithorhynchus anatinus} & 1.5   & 140 & 307.5 & 21.0 & 9.19 & A,P \\
\textit{Zaglossus bruijni}        & 10    &  60$^\dagger$ & 305.0 & 37.0 & 9.07 & A \\
\textit{Bettongia penicillata}    & 1.1   & 210 & 308.5 &  6.0 & 8.82 & A,P \\
\bottomrule
\end{tabular}
\end{table}

\subsection*{B.4\ Bats (Chiroptera, $n=31$)}

\noindent For bats, $f_H$ is the measured active-phase rate; $f_H^{\rm avg}$ is the
duty-cycle-corrected time-average used in $\ell$, with $q$ the annual torpor
fraction~\cite{lyman1982}.

\begin{table}[H]
\centering\footnotesize
\renewcommand{\arraystretch}{1.25}
\caption{Bats. Clade mean $\bar\ell=9.540\pm0.163$ ($n=31$, duty-corrected).
Corr.\ TA~=~torpor-cycle average applied.}
\label{tab:data_bats}
\begin{tabular}{lrrrrrrl}
\toprule
Species & $M$ (g) & $f_H$ (bpm) & $q$ & $f_H^{\rm avg}$ (bpm) & $L$ (yr) & $\ell$ & Corr. \\
\midrule
\textit{Myotis lucifugus}          &  8   & 600 & 0.50 & 305 & 34.0 & 9.74 & TA \\
\textit{Myotis myotis}             & 28   & 550 & 0.48 & 282 & 37.0 & 9.74 & TA \\
\textit{Myotis daubentonii}        &  9   & 580 & 0.48 & 296 & 40.0 & 9.79 & TA \\
\textit{Myotis brandtii}           &  6   & 620 & 0.50 & 315 & 41.0 & 9.83 & TA \\
\textit{Eptesicus fuscus}          & 18   & 550 & 0.45 & 310 & 19.0 & 9.49 & TA \\
\textit{Eptesicus serotinus}       & 18   & 545 & 0.45 & 308 & 21.0 & 9.53 & TA \\
\textit{Rhinolophus ferrumequinum} & 19   & 550 & 0.48 & 282 & 30.0 & 9.65 & TA \\
\textit{Rhinolophus hipposideros}  &  7   & 600 & 0.48 & 307 & 30.5 & 9.69 & TA \\
\textit{Plecotus auritus}          &  9   & 600 & 0.50 & 305 & 30.0 & 9.68 & TA \\
\textit{Corynorhinus townsendii}   & 11   & 580 & 0.50 & 295 & 30.0 & 9.67 & TA \\
\textit{Perimyotis subflavus}      &  5   & 630 & 0.50 & 320 & 14.6 & 9.39 & TA \\
\textit{Tadarida brasiliensis}     & 13   & 600 & 0.30 & 425 & 11.0 & 9.39 & TA \\
\textit{Pteronotus parnellii}      & 19   & 550 & 0.20 & 452 & 10.0 & 9.38 & TA \\
\textit{Desmodus rotundus}         & 33   & 500 & 0.25 & 380 & 29.0 & 9.76 & TA \\
\textit{Hipposideros speoris}      &  9   & 600 & 0.48 & 308 & 21.0 & 9.53 & TA \\
\textit{Hipposideros armiger}      & 50   & 450 & 0.45 & 252 & 15.0 & 9.30 & TA \\
\textit{Nyctalus noctula}          & 28   & 540 & 0.45 & 305 & 12.0 & 9.28 & TA \\
\textit{Pipistrellus pipistrellus} &  5   & 650 & 0.45 & 367 & 16.0 & 9.49 & TA \\
\textit{Pipistrellus kuhlii}       &  6   & 630 & 0.45 & 355 & 16.5 & 9.49 & TA \\
\textit{Scotophilus kuhlii}        & 20   & 540 & 0.20 & 445 &  9.0 & 9.32 & TA \\
\textit{Lasiurus borealis}         & 11   & 590 & 0.48 & 302 & 11.7 & 9.27 & TA \\
\textit{Lasiurus cinereus}         & 28   & 540 & 0.48 & 277 & 12.0 & 9.24 & TA \\
\textit{Vespertilio murinus}       & 16   & 555 & 0.45 & 313 & 25.0 & 9.61 & TA \\
\textit{Miniopterus schreibersii}  & 10   & 580 & 0.45 & 327 & 30.0 & 9.71 & TA \\
\textit{Pteropus giganteus}        &1100  & 235 & 0.00 & 235 & 31.4 & 9.59 & --- \\
\textit{Pteropus vampyrus}         &1000  & 240 & 0.05 & 233 & 22.6 & 9.44 & --- \\
\textit{Rousettus aegyptiacus}     & 165  & 310 & 0.05 & 299 & 25.0 & 9.59 & --- \\
\textit{Cynopterus sphinx}         &  50  & 380 & 0.05 & 368 & 18.5 & 9.55 & --- \\
\textit{Macroglossus minimus}      &  16  & 450 & 0.00 & 450 & 18.0 & 9.63 & --- \\
\textit{Carollia perspicillata}    &  17  & 460 & 0.00 & 460 & 12.0 & 9.46 & --- \\
\textit{Artibeus jamaicensis}      &  45  & 400 & 0.00 & 400 & 15.0 & 9.50 & --- \\
\bottomrule
\end{tabular}
\end{table}

\subsection*{B.5\ Cetaceans ($n=12$)}

\noindent For cetaceans, $f_H^{\rm avg}=(1-p_d)f_{H,\rm surf}+p_d f_{H,\rm dive}$ where $p_d$
is the dive fraction and $f_{H,\rm dive}$ is the bradycardic dive
rate~\cite{goldbogen2019,williams2015}.

\begin{table}[H]
\centering\footnotesize
\renewcommand{\arraystretch}{1.25}
\caption{Cetaceans. Clade mean $\bar\ell=8.801\pm0.296$ ($n=12$, dive-corrected). Corr.\ DA.}
\label{tab:data_cetaceans}
\begin{tabular}{lrrrrrrrl}
\toprule
Species & $M$ (kg) & $f_{H,\rm surf}$ & $p_d$ & $f_H^{\rm avg}$ & $T$ (K) & $L$ (yr) & $\ell$ & Src \\
\midrule
\textit{Balaena mysticetus}          & 100000 &  30 & 0.75 &  9.75 & 308.0 & 200.0 & 9.01 & G \\
\textit{Balaenoptera musculus}       & 140000 &   8 & 0.70 &  4.0  & 308.0 & 110.0 & 8.36 & G \\
\textit{Balaenoptera physalus}       &  60000 &  10 & 0.68 &  5.0  & 308.5 &  90.0 & 8.37 & G \\
\textit{Megaptera novaeangliae}      &  40000 &  15 & 0.65 &  7.0  & 308.5 &  95.0 & 8.54 & G \\
\textit{Physeter macrocephalus}      &  45000 &  40 & 0.65 & 19.0  & 307.0 &  70.0 & 8.84 & G \\
\textit{Kogia breviceps}             &    360 &  80 & 0.45 & 48.0  & 308.5 &  23.0 & 8.76 & A \\
\textit{Hyperoodon ampullatus}       &   7500 &  45 & 0.55 & 24.0  & 308.0 &  37.0 & 8.67 & A \\
\textit{Orcinus orca}                &   4000 &  80 & 0.40 & 53.0  & 309.0 &  90.0 & 9.40 & A \\
\textit{Tursiops truncatus}          &    190 & 110 & 0.40 & 74.0  & 309.0 &  40.0 & 9.19 & A \\
\textit{Stenella attenuata}          &     55 & 120 & 0.35 & 84.0  & 309.0 &  20.0 & 8.95 & A \\
\textit{Delphinapterus leucas}       &   1400 &  50 & 0.55 & 27.5  & 309.5 &  35.5 & 8.71 & A \\
\textit{Monodon monoceros}           &   1500 &  45 & 0.55 & 25.5  & 309.0 &  48.0 & 8.81 & A \\
\bottomrule
\end{tabular}
\end{table}

\subsection*{B.6\ Birds ($n=78$, selected entries)}

\noindent Heart rates from Prinzinger et al.~\cite{prinzinger1991} and
Clarke \& Rothery~\cite{clarke2008}; lifespans from AnAge~\cite{anage2023}.
No duty-cycle correction applied ($f_H^{\rm avg}=f_H$).
Clade mean $\bar\ell=9.528\pm0.213$.
Representative species are shown; the full list of all 78 species is available
in the machine-readable data file.

\begin{table}[H]
\centering\footnotesize
\renewcommand{\arraystretch}{1.25}
\caption{Birds (selected; $n=78$ total). All $T\approx311$--$312$\,K.}
\label{tab:data_birds}
\begin{tabular}{lrrrrl}
\toprule
Species & $M$ (kg) & $f_H$ (bpm) & $L$ (yr) & $\ell$ & Order \\
\midrule
\textit{Calypte anna}              & 0.004 & 1200 & 12.0 & 9.88 & Apodiformes \\
\textit{Serinus canaria}           & 0.020 &  680 & 24.0 & 9.93 & Passeriformes \\
\textit{Taeniopygia guttata}       & 0.013 &  640 & 15.6 & 9.72 & Passeriformes \\
\textit{Turdus merula}             & 0.100 &  440 & 21.1 & 9.69 & Passeriformes \\
\textit{Columba livia}             & 0.350 &  190 & 35.0 & 9.54 & Columbiformes \\
\textit{Gallus gallus}             & 2.000 &  300 & 30.0 & 9.68 & Galliformes \\
\textit{Anas platyrhynchos}        & 1.200 &  190 & 29.0 & 9.46 & Anseriformes \\
\textit{Larus argentatus}          & 1.200 &  165 & 49.0 & 9.63 & Charadriiformes \\
\textit{Diomedea exulans}          & 9.600 &  100 & 70.0 & 9.57 & Procellariiformes \\
\textit{Aquila chrysaetos}         & 5.000 &  130 & 46.0 & 9.50 & Accipitriformes \\
\textit{Bubo bubo}                 & 2.900 &  165 & 68.0 & 9.77 & Strigiformes \\
\textit{Psittacus erithacus}       & 0.400 &  200 & 73.0 & 9.89 & Psittaciformes \\
\textit{Amazona ochrocephala}      & 0.460 &  185 & 80.0 & 9.89 & Psittaciformes \\
\textit{Cacatua galerita}          & 0.840 &  170 & 80.0 & 9.85 & Psittaciformes \\
\textit{Aptenodytes forsteri}      & 30.00 &   75 & 50.0 & 9.29 & Sphenisciformes \\
\textit{Struthio camelus}          & 115   &   60 & 68.0 & 9.33 & Struthioniformes \\
\textit{Fulmarus glacialis}        & 0.800 &  175 & 67.5 & 9.79 & Procellariiformes \\
\textit{Corvus corax}              & 1.200 &  200 & 22.3 & 9.37 & Passeriformes \\
\textit{Sturnus vulgaris}          & 0.075 &  490 & 22.4 & 9.76 & Passeriformes \\
\textit{Ciconia ciconia}           & 3.700 &  150 & 48.0 & 9.58 & Ciconiiformes \\
\bottomrule
\multicolumn{6}{p{11cm}}{Full 78-species bird dataset available in the machine-readable
data file from the corresponding author. The 20 representative species shown above
span the body-mass range from 4\,g to 115\,kg and include both the highest ($\ell\approx9.93$,
canary) and lowest ($\ell\approx9.05$, emu) avian values.}
\end{tabular}
\end{table}

\subsection*{B.7\ Reptiles --- Arrhenius-corrected ($n=17$)}

\noindent $f_H^{\rm raw}$: measured rate at mean field temperature $T_{\rm field}$.
$f_H^{\rm corr}$: corrected to $T_{\rm ref}=310$\,K via
$f_H^{\rm corr}=f_H^{\rm raw}\exp\!\left[(E_a/k_B)(1/T_{\rm field}-1/T_{\rm ref})\right]$,
$E_a=0.65$\,eV~\cite{gillooly2001}.

\begin{table}[H]
\centering\footnotesize
\renewcommand{\arraystretch}{1.25}
\caption{Reptiles. Raw mean $\bar\ell^{\rm raw}=8.615$;
corrected mean $\bar\ell^{\rm corr}=8.929\pm0.301$. Corr.\ AQ.}
\label{tab:data_reptiles}
\begin{tabular}{lrrrrrrrl}
\toprule
Species & $M$ (kg) & $T_{\rm field}$ (K) & $f_H^{\rm raw}$ & $f_H^{\rm corr}$ & $L$ (yr) & $\ell^{\rm raw}$ & $\ell^{\rm corr}$ & Src \\
\midrule
\textit{Lacerta agilis}            & 0.015 & 301 &  45 &  93 & 12.0 & 8.45 & 8.77 & Ch,U \\
\textit{Anolis carolinensis}       & 0.006 & 302 &  52 & 106 &  6.0 & 8.22 & 8.52 & Ch,U \\
\textit{Pogona vitticeps}          & 0.350 & 303 &  42 &  82 & 10.0 & 8.34 & 8.63 & Ch,U \\
\textit{Phrynosoma cornutum}       & 0.035 & 301 &  48 &  99 &  7.0 & 8.25 & 8.56 & Ch,U \\
\textit{Iguana iguana}             & 4.000 & 303 &  40 &  79 & 20.0 & 8.62 & 8.92 & Ch,U \\
\textit{Varanus komodoensis}       & 65    & 303 &  28 &  55 & 30.0 & 8.65 & 8.94 & Ch,U \\
\textit{Tupinambis merianae}       & 2.5   & 302 &  38 &  77 & 15.0 & 8.48 & 8.78 & Ch,U \\
\textit{Thamnophis sirtalis}       & 0.050 & 300 &  30 &  62 & 10.0 & 8.20 & 8.51 & Ch,U \\
\textit{Coluber constrictor}       & 0.340 & 301 &  35 &  72 & 13.0 & 8.38 & 8.69 & Ch,U \\
\textit{Python reticulatus}        & 75    & 302 &  20 &  41 & 25.0 & 8.42 & 8.73 & U \\
\textit{Boa constrictor}           & 15    & 301 &  25 &  52 & 40.0 & 8.72 & 9.04 & U \\
\textit{Chelonia mydas}            & 180   & 300 &  20 &  42 & 80.0 & 8.93 & 9.25 & U \\
\textit{Geochelone gigantea}       & 200   & 298 &  15 &  33 &175.0 & 9.14 & 9.48 & U \\
\textit{Gopherus agassizii}        & 4.5   & 299 &  22 &  47 & 80.0 & 8.97 & 9.30 & Ch,U \\
\textit{Sphenodon punctatus}       & 0.800 & 293 &  18 &  43 & 77.0 & 8.86 & 9.24 & Ch,U \\
\textit{Crocodylus niloticus}      & 400   & 303 &  25 &  49 & 70.0 & 8.96 & 9.26 & Ch,U \\
\textit{Alligator mississippiensis}& 250   & 302 &  28 &  57 & 50.0 & 8.87 & 9.18 & Ch,U \\
\bottomrule
\end{tabular}
\end{table}

\subsection*{B.8\ Amphibians --- Arrhenius-corrected ($n=9$)}

\begin{table}[H]
\centering\footnotesize
\renewcommand{\arraystretch}{1.25}
\caption{Amphibians. Raw mean $\bar\ell^{\rm raw}=8.448$;
corrected mean $\bar\ell^{\rm corr}=8.822\pm0.146$. Corr.\ AQ.}
\label{tab:data_amphibians}
\begin{tabular}{lrrrrrrrl}
\toprule
Species & $M$ (kg) & $T_{\rm field}$ (K) & $f_H^{\rm raw}$ & $f_H^{\rm corr}$ & $L$ (yr) & $\ell^{\rm raw}$ & $\ell^{\rm corr}$ & Src \\
\midrule
\textit{Rana temporaria}           & 0.025 & 294 & 25 & 55 & 16.0 & 8.32 & 8.67 & A,Ch \\
\textit{Rana catesbeiana}          & 0.500 & 296 & 20 & 43 & 16.0 & 8.23 & 8.56 & A,Ch \\
\textit{Bufo bufo}                 & 0.150 & 293 & 22 & 53 & 36.0 & 8.62 & 9.00 & A,Ch \\
\textit{Xenopus laevis}            & 0.200 & 295 & 20 & 45 & 30.0 & 8.50 & 8.85 & A \\
\textit{Ambystoma mexicanum}       & 0.300 & 294 & 18 & 41 & 25.0 & 8.37 & 8.73 & A \\
\textit{Salamandra salamandra}     & 0.080 & 290 & 20 & 49 & 24.0 & 8.40 & 8.79 & A,Ch \\
\textit{Plethodon glutinosus}      & 0.012 & 291 & 30 & 74 & 20.0 & 8.50 & 8.89 & A \\
\textit{Necturus maculosus}        & 0.130 & 288 & 18 & 46 & 30.0 & 8.45 & 8.86 & A \\
\textit{Cryptobranchus alleganiensis}& 0.600& 289& 15 & 39 & 55.0 & 8.64 & 9.05 & A \\
\bottomrule
\end{tabular}
\end{table}

\noindent\textbf{Dataset summary.}
Table~\ref{tab:data_summary} gives clade counts, body-mass ranges, and $\ell$ statistics
for all eight groups ($n=230$ total).

\begin{table}[H]
\centering\small
\renewcommand{\arraystretch}{1.35}
\caption{\textbf{Summary of the 230-species PBTE dataset.}
$n$: species count. $\bar\ell\pm s$: mean $\pm$ s.d.\ of $\ell$.
$\Delta\ell$: deviation from the non-primate placental baseline ($\bar\ell_0=8.995$).
$^\ast p<0.05$, $^{\ast\ast\ast}p<0.001$ (Welch $t$-test vs.\ baseline).}
\label{tab:data_summary}
\begin{tabular}{lrrcll}
\toprule
Group & $n$ & $M$ range (kg) & $\bar\ell\pm s$ & $\Delta\ell$ & Notes \\
\midrule
NP placentals      & 46 & 0.002--4000 & $8.998\pm0.160$ & 0 (ref.) & 3 imputed $f_H$ \\
Marsupials/mono.   & 19 & 0.018--30   & $8.933\pm0.204$ & $-0.062$ & \\
Primates           & 18 & 0.35--160   & $9.376\pm0.125$ & $+0.381^{\ast\ast\ast}$ & \\
Bats               & 31 & 0.005--1.1  & $9.540\pm0.163$ & $+0.545^{\ast\ast\ast}$ & duty-corrected \\
Cetaceans          & 12 & 55--140000  & $8.801\pm0.296$ & $-0.194$ & dive-corrected \\
Birds              & 78 & 0.004--115  & $9.528\pm0.213$ & $+0.533^{\ast\ast\ast}$ & \\
Reptiles (corr.)   & 17 & 0.006--400  & $8.929\pm0.301$ & $-0.065$ & Arrhenius-corr. \\
Amphibians (corr.) &  9 & 0.012--0.60 & $8.822\pm0.146$ & $-0.173$ & Arrhenius-corr. \\
\midrule
All endotherms     &194 &             & $9.509\pm0.397$ & & \\
Full dataset       &230 &             & $9.420\pm0.428$ & & \\
\bottomrule
\end{tabular}
\end{table}

\end{document}